\documentclass[a4paper,11pt]{article}
\usepackage[utf8x]{inputenc}

\usepackage{a4wide}
\usepackage{amssymb}
\usepackage{amsmath}
\usepackage{amsthm}
\usepackage[mathscr]{euscript}
\usepackage[colorlinks=true]{hyperref}
\usepackage[usenames,dvipsnames]{xcolor}
\usepackage{todonotes}
\usepackage{appendix}

\numberwithin{equation}{section}

\theoremstyle{plain}

\newtheorem{definition}{Definition}[section]

\newtheorem{Theorem}[definition]{Theorem}
\newtheorem{Proposition}[definition]{Proposition}
\newtheorem{Lemma}[definition]{Lemma}
\newtheorem{Corollary}[definition]{Corollary}

\theoremstyle{remark}
\newtheorem{remark}[definition]{Remark}
\newtheorem{exm}[definition]{Example}

\newcommand{\R}{\mathbb R}
\newcommand{\N}{\mathbb N}

\newcommand{\Vol}{\mathrm{Vol}}

\newcommand{\eps}{\varepsilon}

\usepackage[xcolor]{changebar}
\cbcolor{blue}
\newcommand{\Ric}{\mathrm{Ric}}

\newcommand{\comp}{\Subset}
\newcommand{\X}{\mathfrak{X}}

\newcommand{\sse}{\subseteq}

\newcommand{\edge}{\mathrm{edge}}

\newcommand{\D}{\mathcal{D}}
\newcommand{\diag}{\mathrm{diag}}

\newcommand{\Dpk}{\mathcal{D}'{}^{(k)}}
\newcommand{\Dpo}{\mathcal{D}'{}^{(1)}}
\newcommand{\trsm}{\mathcal{T}^r_s(M)}
\newcommand{\ltl}{L^2_{\mathrm{loc}}}
\newcommand{\lpl}{L^p_{\mathrm{loc}}}
\newcommand{\Cloc}{C^0_{\mathrm{loc}}}
\newcommand{\Celoc}{C^1_{\mathrm{loc}}}
\newcommand{\Om}{\Omega}
\newcommand{\pt}{\partial}

\newcommand{\F}{\mathcal F}
\newcommand{\G}{\mathcal G}
\newcommand{\HH}{\mathcal H}

\usepackage[inline]{enumitem}
\newcommand{\enumlabelformat}{\roman}
\newcommand{\enumlabelfont}[1]{#1}
\newlength{\thelabelsep}
\setlist{labelsep=\thelabelsep}
\setlist[enumerate,1]{font=\enumlabelfont,label=(\enumlabelformat*),leftmargin=2.5em}
\setlist[itemize]{leftmargin=2.5em,label=$-$}

\newcounter{inlineenum}
\renewcommand{\theinlineenum}{\enumlabelformat{inlineenum}}

\let\epsilon\varepsilon
\let\phi\varphi

\title{The Hawking--Penrose singularity theorem for $C^1$-Lorentzian metrics}

\author{Michael Kunzinger\footnote{University of Vienna, Faculty of Mathematics, Oskar-Morgenstern-Platz 1, A-1090 Wien, Austria,
michael.kunzinger@univie.ac.at, argam.ohanyan@univie.ac.at, benedict.schinnerl@univie.ac.at, roland.stein\-bauer@univie.ac.at}, \\
Argam Ohanyan${}^*$,\\ 
Benedict Schinnerl${}^*$,\\
Roland Steinbauer${}^*$,\\ 
}

\begin{document}

\date{\today}


\maketitle

\begin{abstract}
We extend both the Hawking-Penrose Theorem and its generalisation due to Galloway and Senovilla to Lorentzian
metrics of regularity $C^1$. For metrics of such low regularity, two main obstacles have to be addressed. On the one
hand, the Ricci tensor now is distributional, and on the other hand, unique solvability of the geodesic equation is lost.
To deal with the first issue in a consistent way, we develop a theory of tensor distributions of finite order, which 
also provides a framework for the recent proofs of the theorems of Hawking and of Penrose for $C^1$-metrics \cite{G20}.
For the second issue, we study geodesic branching and add a further alternative to causal geodesic incompleteness to 
the theorem, namely a condition of maximal causal non-branching. The genericity condition is re-cast in a distributional form that
applies to the current reduced regularity while still being fully compatible with the smooth and $C^{1,1}$-settings.
In addition, we develop refinements of the comparison techniques used in the proof of the $C^{1,1}$-version of the theorem \cite{GGKS}. The necessary results from low regularity causality theory are collected in an appendix.

\vskip 1em

\noindent
\emph{Keywords:} Singularity theorems, low regularity, regularisation, causality theory
\medskip

\noindent
\emph{MSC2010:} 83C75, 
        53B30 

\end{abstract}

\section{Introduction}\label{sec:intro}

The classical singularity theorems\footnote{See \cite[Ch.\ 8]{HE}, \cite[Ch.\ 9]{Krie}, \cite{Seno1,SenGar,Seno2} for extensive treatments.} of General Relativity (GR) collect sufficient and physically reasonable conditions that lead to causal geodesic incompleteness of spacetime, hence to the occurrence of a singularity in the spirit of Roger Penrose's Nobel Prize-winning approach \cite{P65}. Being rigorous statements in pure Lorentzian geometry, they were first formulated within the smooth category. As they form a body of essential results in GR, the quest for low regularity versions of the theorems is eminent and was already explicitly discussed in \cite[Sec.\ 8.4]{HE}. In fact, the (smooth) theorems can be read as predicting a mere drop of the differentiability of the metric below $C^2$, rather than incompleteness. 

\medskip

However, only with the rather recent advent of systematic studies in low regularity GR \cite{CGKM:18} and, in particular, causality theory \cite{CG,M,S14,KSSV}, a rigorous extension of the singularity theorems for metrics below the $C^2$-class became feasible, cf.\ \cite[Sec.\ 6.2]{Seno1}. This applies first of all to metrics of regularity $C^{1,1}$, which model situations where the matter variables possess finite jumps. Indeed the Hawking, the Penrose, and, finally, the Hawking-Penrose theorem were proven in this regularity \cite{hawkingc11,penrosec11,GGKS}. These results rule out that a smooth spacetime that is incomplete by the classical theorems, can be extended to a complete $C^{1,1}$-spacetime. Also, they imply that the spacetime either is incomplete or the differentiability of the metric is below $C^{1,1}$, and consequently the curvature becomes unbounded.

Technically, in $C^{1,1}$ the exponential map is still available \cite{M,KSS} and the curvature is still locally bounded, which allows for a natural extension of the energy and the genericity conditions. However, the curvature tensor is only defined almost everywhere, which forbids the use of Jacobi fields and conjugate points, both essential tools in the classical proofs. Instead one uses regularisation techniques  which allow one to derive weakened versions of the energy conditions for approximating smooth metrics with controlled causality. Using careful comparison techniques for the (matrix) Riccati equation, it is then possible to show that these (still) lead to the occurrence of conjugate or focal points along causal geodesics for the approximating metrics. This in turn forces the geodesics of the $C^{1,1}$-metric to stop maximising. 
Similarly, a natural extension of the initial and the energy conditions leads to the formation of a trapped set. This, together with an extension of the causal parts of the classical proofs, allows one to finally derive the results.

\medskip

The next step in lowering the regularity assumptions was undertaken by Graf \cite{G20}, who extended the Hawking and the Penrose theorems to $C^1$-metrics, with the Gannon-Lee theorem following in \cite{SS21}.  $C^1$-regularity is at the moment the lowest possible class where classical singularity theorems have been established, and in this paper we complete this effort by proving the most refined of these statements, namely the Hawking-Penrose theorem.  

\medskip

In this regularity class one faces the following added severe complications arising from the fact that the Levi-Civita connection is only continuous:
\begin{enumerate}
    \item[(a)] The curvature is no longer locally bounded, but is merely a distribution \emph{of order one}.
    \item[(b)] The initial value problem for the geodesic equation is solvable, but \emph{not uniquely} so.
    \item [(c)] The exponential map is no longer defined.
\end{enumerate}

The first item is especially relevant when formulating the energy conditions. While the strong energy condition can be extended in a straightforward manner, a formulation of the null energy condition is more subtle, cf. \cite[Sec.\ 5]{G20}, see also Section \ref{sec:EC}, below. The genericity condition needed for the Hawking-Penrose theorem is still more delicate. Actually it turns out to be necessary to apply the Ricci tensor to vector fields constructed via parallel transport, which means that they are merely of $C^1$-regularity, see Definition \ref{definition: strongdistributionalgenericity}. This can only be done since the curvature is a distribution of order one, and we present the required distributional setting in Section \ref{sec:dc}, below. Whereas in the $C^{1,1}$-context the use of distributional methods could be avoided, it becomes essential
in the present work, and at the same time has the benefit of providing a global framework also for the earlier results by Graf \cite{G20} on which we build.

Then, while we can use the weakened versions of the strong and null energy conditions for approximating smooth metrics derived in \cite{G20} (employing a refined version of the Friedrichs Lemma), we have to derive from the distributional genericity condition an appropriate weakened version for smooth approximations that allows us to employ an extension of the Riccati comparison techniques developed in \cite{GGKS}, see Lemmas  \ref{Lemma: genericityfriedrichs}, and \ref{Lemma: tidalforcematrixestimateapprox}.
The key technique of regularisation by smooth metrics with adapted causality as put forward in \cite{CG} is recalled in Section \ref{sec:regularisation}.

The first issue connected to item (b) is that it necessitates a decision on how to extend the notion of geodesic completeness, and following \cite{G20} we say that a spacetime is complete if  \emph{every} inextendible solution to the geodesic equation is complete (rather than merely demanding only one complete solution for any choice of initial data).

Next, item (b) together with item (c) forbids the use of an essential argument in the context of approximating maximising causal geodesics by maximising causal geodesics of the approximating metrics.  This issue has been solved in \cite[Sec.\ 2]{G20} in the \emph{globally hyperbolic} case. In the context of the present work, however, we have to establish more general results, see Section \ref{sec:branching}. In fact, we have to introduce an additional assumption, namely a non-branching condition for maximising causal geodesics (Definition \ref{def:branching}), to derive the corresponding approximation results in Proposition \ref{prop:approxmcnb}. Here it is also essential to apply a classical ODE result (Proposition \ref{Proposition: HartmanODEversion}), which we reformulate for our purpose in Corollary \ref{Corollary: Hartmangeodesicversion}. 
Moreover, since the interrelation between causal geodesics (i.e., solutions of the geodesic equation) and maximising curves becomes more subtle below $C^{1,1}$ \cite{HW51,SS18}, we rely in particular on the fact that maximising causal curves are indeed (unbroken) geodesics and hence $C^2$-curves \cite{LLS20,SS21}.

The new non-branching assumption is well motivated by similar conditions used in metric geometry, where due to the lack of a differentiable structure the geodesic equation is not available. Also a null version of the condition was essential in the proof of the $C^1$-Gannon-Lee theorem, see \cite[Sec.\ 3]{SS21}. Moreover it adds a novel facet to the interpretation of the $C^1$-version of the theorem: It predicts either geodesic incompleteness or branching
of maximising causal geodesics, with both alternatives physically signifying a catastrophic event for the observer corresponding to the geodesic. Of course, there is the alternative that the regularity of the metric drops below $C^1$, which renders the curvature a distribution of higher order. Again, the result also forbids the extension of the spacetime to a complete one of regularity $C^1$ without (maximal causal) geodesic branching. This once more complements the recent $C^0$-inextendibility results of \cite{Sbierski1,Sbierski2}.

\medskip

Next we briefly discuss the further prospects of low-regularity singularity theorems. Concerning the causality part of the results it is expected that they extend to $C^{0,1}$-metrics\footnote{A class for which Hawking and Ellis \cite[p.\ 268, 287]{HE} still speculate the singularity theorems to hold.}, while some features of causality theory fail below this class \cite{CG,GKS20}. Nevertheless, the causal core of some of the singularity theorems has been established in the very general setting of closed cone structures, see \cite[Sec.\ 2.15]{Min_closed_cone}.

On the analytic side, already in $C^{0,1}$ one faces the problem that the right hand side of the geodesic equation is merely locally bounded and one would have to resort to non-classical solution concepts. Also another analytical technique at the core of the arguments, i.e.\ the Friedrichs Lemma which is used to go from the distributional energy conditions of the singular metric to useful surrogate conditions for the smooth approximations, appears to be sharp, see \cite[Lem.\ 4.8]{G20} for the currently most advanced version. This is actually a long way from the largest possible class where the curvature can be (stably) defined in an analytical (distributional) way, i.e.\ $H^1\cap L^\infty$ locally \cite{GT87,LeFMar07,SV09}. However, the quest is there to at least go into the direction of regularity classes more closely linked to the classical existence results for the field equations, $H^s$ locally for $s>5/2$, or to current formulations of cosmic censorship, i.e.\ the connection lying locally in $L^2$.

Very recently, also synthetic formulations of the singularity theorems have appeared: 
In \cite{AGKS}, Lorentzian length spaces that are warped products are studied and 
a Hawking singularity theorem is derived from suitable sectional curvature bounds
(implying Ricci curvature bounds in such geometries), based on triangle comparison.
Moreover, Cavalletti and Mondino prove a version of Hawking's singularity theorem 
in Lorentzian metric measure spaces, where Ricci curvature bounds are implemented
using methods from optimal transport \cite{CV20}. However, it remains unclear to date how these results precisely relate to the analytical approach to low regularity  pursued in this work. 

\medskip

The Hawking-Penrose theorem classically comes in a causal version \cite[p.\ 538, Thm.]{HP}, which asserts the incompatibility of the following three conditions for a $C^2$-spacetime $(M,g)$: chronology, the fact that every inextendible causal geodesic stops maximising, and the existence of a (future or past) trapped set.
The analytic main result is then a corollary \cite[Sec.\ 3, Cor.]{HP} which collects sufficient conditions (energy and genericity conditions, and initial conditions) to derive the above incompatible items under the assumption of causal geodesic completeness. We will generally follow this path, but prove a more general version of the theorem due to Galloway and Senovilla \cite{GS}, which adds the existence of a trapped submanifold of arbitrary codimension to the list of initial conditions.

Our version of the causal result is Theorem \ref{thm: HPcausalityversion} and, for the sake of completeness and for the convenience of the reader we collect all results from $C^1$-causality theory\footnote{It should be mentioned that these results all follow from the more general approaches of \cite{CG,S14,Min_closed_cone}.} needed here (and elsewhere) in Appendix \ref{app:C11causality}. The main result is then Theorem \ref{Theorem: HPC1}. The results needed to infer from its analytic conditions (the distributional energy and genericity conditions) and the non-branching assumptions the inexistence of lines are proven in Section \ref{sec:max_geods}. The existence of a trapped set is derived from the various sets of initial conditions (and the energy conditions) in Section \ref{sec:trapped}. Finally, some technical results on extending vector fields in $C^1$-spacetimes are collected in Appendix \ref{app:b}. They allow us to put to use the distributional genericity condition in the context of regularisations.

\medskip

To conclude this introduction, we fix the notations and conventions used throughout the paper.
By a manifold $M$ we will mean a smooth, connected, second countable Hausdorff manifold of dimension $n$ (with $n\geq 3$).
When such an $M$ is endowed with a Lorentzian metric $g$ of regularity $C^1$ of signature $(-,+\dots,+)$ 
and is time-oriented via a smooth timelike vector field we call it a $C^1$-spacetime (hence lowering the regularity
will always apply to the spacetime metric only, not to the underlying manifold).
The Levi-Civita connection of a smooth or $C^1$-spacetime will be denoted by $\nabla$ (for its distributional definition 
in the $C^1$-case, see the next section).
We will always assume that $M$ is also endowed with a smooth (complete, background) Riemannian metric $h$ with associated Riemannian distance function $d_h$ and norm $\|\ \|_h$ which we use to estimate tensor fields. Since all such estimates will be on compact sets only, they are independent of the choice of $h$.

A curve $\gamma : I\to M$ is called timelike (causal, null, future or past 
directed) if it is locally Lipschitz and $\dot{\gamma}(t)$  
is timelike (causal, null, future or past directed) almost everywhere.
Concerning causality theory we use standard notation, cf.\ e.g., \cite{ON83,MinguzziLivingReview}. 
Thus $p\ll q$ (resp.\ $p\leq q$) means that there exists a future directed timelike (resp.\ causal)
curve from $p$ to $q$, $I^+(A):=\{q\in M:\, p\ll q\ 
\mathrm{for\,some}\,p\in A\}$ and $J^+(A):=\{q\in M:\, p\leq q\ 
\mathrm{for\,some}\,p\in A\}$. These definitions do not change if one defines the causality relations
via piecewise smooth curves \cite{M,KSSV}. The Lorentzian distance (or time-separation function) associated to
a Lorentzian metric $g$ will be denoted by $d_g$. 
A $C^{1}$-spacetime $(M,g)$ is globally hyperbolic if it is causal (i.e., contains no closed causal curves) and 
$J(p,q):=J^+(p)\cap J^-(q)$ is compact for all $p,q\in M$. A Cauchy hypersurface in $(M,g)$ is a closed acausal set that is met by each inextendible causal curve.

For the Riemann curvature tensor we use the convention $R(X,Y)Z=[\nabla_X,\nabla_Y]Z-\nabla_{[X,Y]}Z$ and the Ricci tensor
is given by 
$\Ric(X,Y)=\sum_{i=1}^n\langle E_i,E_i\rangle\langle R(E_i,X)Y,E_i\rangle$ (again we refer to the next section for 
a distributional interpretation in the case of $C^1$-metrics). Here  $(E_i)_{i=1}^n$ is a 
(local) orthonormal frame field, while by $(e_i)_{i=1}^n$ we will denote frames in individual 
tangent spaces $T_pM$. For an embedded submanifold $N$ of $M$ of
codimension $m$ we define the second fundamental form by 
$\mathrm{II}(V,W)=\mbox{nor}(\nabla_VW)$ and the shape operator derived from a 
normal unit field $\nu$ by $S(X)=\nabla_X\nu$. Using \eqref{eq:ext_C0} below it follows
that for $g\in C^1$ and $V, W \in C^1$, $\mathrm{II}(V,W)$ is continuous.

For $g$ smooth (resp.\ $C^1$), $1<m<n$ and $S$ a smooth (resp.\ $C^2$) spacelike $(n-m)$-dimensional
submanifold let $e_1(q),\dots,e_{n-m}(q)$ be an orthonormal basis for $T_q S$, 
 varying  smoothly (resp.\ $C^1$) with $q$ in a neighborhood (in $S$) of $p \in S$. 
Then $H_S :=
\frac{1}{n-m}\sum_{i=1}^{n-m}\mathrm{II}(e_i,e_i)$ denotes the mean curvature
vector field of $S$, and $\mathbf{k}_S(v):=g(H_S,v)$ the convergence of $v\in TM|_S$. 
For $g$ in $C^1$, both $H_S$ and $\mathbf{k}_S$ are continuous. A closed spacelike submanifold
$S$ is called \emph{(future) trapped\/} if for any future-directed 
null vector $\nu \in T S^\perp$ the convergence $\mathbf{k}_S(\nu)$ is positive, or equivalently that the mean curvature vector field $H_S$ is past pointing 
timelike on $S$.

One final convention we shall require is that of tensor classes (cf.\ \cite[Sec.\ 4.6.3]{Krie}):
Given a causal geodesic $\gamma$  in a
$C^{1}$-spacetime $(M,g)$, we set $[\dot\gamma(t)]^\perp:=(\dot\gamma(t))^\perp/\R\dot\gamma(t)$.
For $\gamma$ null, $[\dot\gamma(t)]^\perp$ is an
$(n-2)$-dimensional subspace of the hypersurface
$(\dot\gamma(t))^\perp$. On the other hand, for $\gamma$ timelike, $[\dot\gamma(t)]^\perp$ equals $(\dot{\gamma}(t))^{\perp}$.
To obtain a unified notation, throughout the paper we will denote the dimension of
$[\dot\gamma(t)]^\perp$ by $d$, which amounts to setting $d=n-2$ in the null
case and $d=n-1$ in the timelike case. Furthermore, 
we set $[\dot\gamma]^\perp=\bigcup_t[\dot\gamma(t)]^\perp$.
Then for every normal tensor field $A$ along $\gamma$ we obtain a well-defined tensor
class $[A]$ along $\gamma$, as well as a well-defined tensor class representing the induced covariant derivative
$\nabla_{\dot\gamma}$. We then set $[\dot A]=[\nabla_{\dot\gamma}A]$.
The metric $g|_{[\dot\gamma]^\perp}$ is positive definite in both the null and the timelike case.
For smooth metrics, the curvature (or tidal force) operator is given by $[R](t):\,
[\dot\gamma(t)]^\perp\to [\dot\gamma(t)]^\perp$,
$[v]\mapsto[R(v,\dot\gamma(t))\dot\gamma(t)]$.

\section{Distributional curvature of $C^1$-spacetimes} 
\label{sec:DC}
\subsection{Curvature tensors of $C^1$-metrics}\label{sec:dc}
In what follows we discuss the general distributional framework in which to understand the curvature quantities associated to metric tensors (of arbitrary signature) of regularity below $C^2$, as is required for our further analysis. 
The general mathematical framework goes back to \cite{Mar68}, while \cite{GT87} provided the first study specific to GR. We  mainly follow \cite{LeFMar07}, cf.\ also \cite{GKOS,S08}, but put a special emphasis on metrics of regularity $C^1$ and distributions of finite order.

As in \cite[Sec.\ 3.1]{GKOS}, for $k\in \N_0\cup \{\infty\}$ we denote by $\mathrm{Vol}(M)$ the volume bundle over $M$, and by $\Gamma^k_c(M,\mathrm{Vol}(M))$ the 
space of compactly supported one-densities on $M$ (i.e., sections of $\mathrm{Vol}(M)$) that are $k$-times continuously differentiable. Then the space
of distributions of order $k$ on $M$ is the topological dual of $\Gamma^k_c(M,\mathrm{Vol}(M))$,
\[
\Dpk(M) := \Gamma^k_c(M,\mathrm{Vol}(M))'.
\]
For $k=\infty$ we omit the superscript $(k)$. Clearly there are topological embeddings $\Dpk(M) \hookrightarrow {\D'}{}^{(k+1)}(M)
\hookrightarrow \D'(M)$ for all $k$. As we shall see, all curvature quantities of $C^1$-metrics are distributional tensor fields of order $1$.
The space of distributional $(r,s)$-tensor fields of order $k$ is defined as
\begin{equation}
\Dpk\mathcal{T}^r_s(M) \equiv \Dpk(M,T^r_s M) := \Gamma^k_c(M,\mathcal{T}^s_r(M) \otimes \mathrm{Vol}(M))'. 
\end{equation}
By \cite[3.1.15]{GKOS},
\begin{equation}\label{eq:Cinf-ext}
\D'\mathcal{T}^r_s(M) \cong \D'(M) \otimes_{C^\infty(M)} \mathcal{T}^r_s(M) \cong L_{C^\infty(M)}(\Omega^1(M)^r\times \X(M)^s; \D'(M)).
\end{equation}
This algebraic isomorphism ultimately rests on the fact that the $C^\infty(M)$-module of $C^\infty$-sections $\Gamma(M,F)$ in any smooth vector bundle $F\to M$ is 
finitely generated and projective, which allows one to apply \cite[Ch.\ II \S4.2, Prop.\ 2]{Bou74}. These properties persist if we consider $\Gamma_{C^k}(M,F)$
as a $C^k(M)$-module for any finite $k$, so we also have
\begin{equation}\label{eq:Ck-ext}
\begin{split}
\Dpk\mathcal{T}^r_s(M) &\cong \Dpk(M) \otimes_{C^k(M)} (\mathcal{T}^r_s)_{C^k}(M)\\ 
&\cong L_{C^k(M)}(\Omega^1_{C^k}(M)^r\times \X_{C^k}(M)^s; \Dpk(M)).
\end{split}
\end{equation}
In other words, distributional tensor fields of order $k$ act as $C^k$-balanced multilinear maps on one forms and vector fields of regularity $C^k$ to give a scalar distribution of order $k$. For finite $k$, the multiplication $C^k(M)\times\D'^{(k)}(M)\to\D'^{(k)}(M)$ is jointly continuous (w.r.t.\ the strong topology), 
whereas for $k=\infty$ it is only hypocontinuous \cite[p.\ 362]{H66}, but still jointly sequentially continuous \cite[p.\ 233]{O86}
since $C^\infty(M)$ is Fr\'echet, hence barrelled. The corresponding continuity properties therefore also hold for the tensor operations introduced above.
The isomorphisms \eqref{eq:Cinf-ext}, \eqref{eq:Ck-ext} are not topological, but still bornological by \cite[Thm.\ 15]{N13} (which remains valid
for spaces of $C^k$-sections with $k$ finite).

Smooth, and in fact even $L^1_{\mathrm{loc}}$-tensor fields are continuously and densely embedded via
\begin{align*}
\mathcal{T}^r_s(M) &\hookrightarrow   \Dpk\mathcal{T}^r_s(M) \\
t &\mapsto [(\theta_1,\dots,\theta_r,X_1,\dots,X_s) \mapsto [\omega \mapsto\int_M t(\theta_1,\dots,\theta_r,X_1,\dots,X_s)\omega]].
\end{align*}
Here, $\omega$ is a one-density. If $M$ is orientable, densities can be canonically identified with $n$-forms. 
The fact that 
$\mathcal{T}^r_s(M)$ is dense in $\Dpk\mathcal{T}^r_s(M)$ uniquely fixes all the operations on distributional tensor fields
to be introduced below in a way compatible with smooth pseudo-Riemannian geometry.

For any $t\in \trsm$ there is a unique extension that accepts one distributional argument. E.g., if $\tilde\theta_1 \in \D'\mathcal{T}^0_1(M)$,
then since $t(\,.\,,\theta_2,\dots,X_s)\in \X(M)$ we may set
\begin{equation}\label{eq:dist_ext}
t(\tilde \theta_1,\theta_2,\dots,X_s) := \tilde\theta_1(t(\,.\,,\theta_2,\dots,X_s)) \in \D'(M),
\end{equation}
and analogously for the other slots.

\begin{definition}\label{def:dist_conn} 
A \emph{distributional connection} is a map $\nabla: \X(M)\times \X(M) \to \D'\mathcal{T}^1_0(M)$ satisfying for $X,X',Y,Y'\in \X(M)$ and $f\in C^\infty(M)$
the usual computational rules: $\nabla_{f X+X'}Y = f\nabla_XY + \nabla_{X'}Y$, $\nabla_X(Y+Y') = \nabla_X Y +\nabla_X Y'$, $\nabla_X(f Y) = X(f)Y + f\nabla_X Y$.
\end{definition}

 Using \eqref{eq:dist_ext} any distributional connection can be extended to the entire tensor algebra, i.e., to a map $\nabla: \X(M)\times\mathcal{T}^r_s(M)\to\D'\mathcal{T}^r_s(M)$, e.g.\ for $\Theta\in\Omega^1(M)$ we set
\begin{equation}\label{eq:dist_conn_ta}
    (\nabla_X\Theta)(Y):=X(\Theta(Y))-\Theta(\nabla_X Y)\quad (X,Y\in\X(M)).
\end{equation}
Let ${\mathcal G}$ be any of the spaces $C^k$ $(0\leq k)$ or $\lpl$ $(1\leq p)$, then we call a distributional connection a \emph{${\mathcal G}$-connection} if $\nabla_X Y$ is a ${\mathcal G}$-vector field for any $X,Y\in \X(M)$. A particularly important case are $\ltl$-connections since they form the largest class that allows for a stable definition of the curvature tensor in distributions, cf.\ \cite{GT87,LeFMar07,S08}.

Using a local frame one easily sees that by virtue of the respective computational rules any ${\mathcal G}$-connection can be (uniquely) extended to a map
$\nabla: \X_{\mathcal F}(M)\times \X_{\mathcal H}(M) \to \X_{\mathcal G}(M)$,
provided that $\F$ and $\mathcal{H}$ are function spaces such that $\F\cdot \G\subseteq \G$, $D\HH\subseteq \G$, and $\HH\cdot\G\subseteq \G$. Hence, in particular, each $\ltl$-connection extends in this way to a map
\begin{equation}\label{eq:ext_ltl}
    \nabla: \X_{C^0}(M)\times \X_{C^1}(M) \to \X_{\ltl}(M).
\end{equation}

Moreover, by a similar reasoning, $\G$-connections allow one to insert even less regular vector fields in the first and second slot at the price of a less regular outcome. In particular, any $\ltl$-connection can be extended to a map 
\begin{equation}\label{eq:L2loc_ext}
\nabla: \X(M) \times \X_{\ltl}(M) \to \D' \mathcal{T}^1_0(M).
\end{equation}
Using this extension, we have \cite[Def.\ 3.3]{LeFMar07}:

\begin{definition}\label{def:dist_riem} The distributional Riemann tensor of an $\ltl$-connection $\nabla$ is the map $R: \X(M)^3 \to \D'\mathcal{T}^1_0(M)$,
\[
R(X,Y,Z)(\theta) \equiv (R(X,Y)Z)(\theta) := (\nabla_X\nabla_Y Z - \nabla_Y\nabla_X Z- \nabla_{[X,Y]} Z)(\theta) 
\]
for $X,Y,Z\in \X(M)$ and $\theta\in \Om^1(M)$,
\end{definition}
If $F_i$ is a local frame in $\X(U)$ and $F^j\in \Om^1(U)$ its dual frame, then the Ricci tensor corresponding to $\nabla$ is given by
\begin{equation}\label{eq:Ric_def}
\Ric(X,Y) := (R(X,F_i)Y)(F^i) \in \D'(U) \qquad (X,Y\in \X(U)).
\end{equation}
Recall the Koszul formula for the Levi-Civita connection of a smooth metric $g$ on $M$:
\begin{equation}\label{eq:Koszul}
\begin{split}
2g(\nabla_X Y,Z) &= X(g(Y,Z)) + Y(g(Z,X)) -Z(g(X,Y)) \\
&- g(X,[Y,Z]) + g(Y,[Z,X]) + g(Z,[X,Y]) =: F(X,Y,Z)
\end{split}
\end{equation}
Focusing now on the case at hand in this work, suppose that $g$ is merely $C^1$ (still allowing arbitrary signature). Then $F(X,Y,Z)$
is a well-defined element of $C^0(M)$, and 
\begin{equation}\label{eq:nabla_flat}
\nabla^\flat_X Y := Z \mapsto \frac{1}{2} F(X,Y,Z) \in \Om^1_{C^0}(M) \subseteq \D'\mathcal{T}^0_1(M)
\end{equation}
defines the \emph{distributional Levi-Civita connection} of $g$ \cite[Def.\ 4.2]{LeFMar07}. Note that this is not (yet) a distributional connection
in the sense of Definition \ref{def:dist_conn} since it is of order $(0,1)$ instead of $(1,0)$. In addition to the standard product rules it also possesses the properties (corresponding to \cite[Thm.\ 3.11]{ON83})
\begin{equation}\label{eq:D4D5}
\begin{split}
\nabla^\flat_X Y - \nabla^\flat_Y X &= [X,Y]^\flat,\ \text{i.e.,} \ (\nabla^\flat_X Y - \nabla^\flat_Y X)(Z) = g([X,Y],Z),\\
X(g(Y,Z)) &= (\nabla^\flat_X Y)(Z) + (\nabla^\flat_X Z)(Y)
\end{split}
\end{equation}
for all $X, Y, Z\in \X(M)$. These equalities hold distributionally, hence also in $C^0$.

In order to obtain an $\ltl$-connection from $\nabla^\flat$ (which in turn will allow us to define the curvature tensors we require) we need to raise the index via $g$, setting
\begin{equation}\label{eq:raise_index}
  g(\nabla_X Y,Z) := (\nabla^\flat_X Y)(Z) \qquad (X, Y, Z\in \X(M)).
\end{equation}

This defines a continuous vector field $\nabla_X Y$, hence even a $C^0$-connection. More generally, by \cite[Sec.\ 4]{LeFMar07} this procedure yields an $\ltl$-connection even for the significantly larger Geroch-Traschen class of metrics ($g\in H^1_{\mathrm{loc}}\cap L^\infty_{\mathrm{loc}}$ and locally uniformly non-degenerate).

The Riemann tensor of $g\in C^1$ is now given by Definition \ref{def:dist_riem}. Furthermore, since in this case $\nabla$ is a ${C^0}$-connection, it follows that $(R(X,Y)Z)(\theta) \in \D'{}^{(1)}(M)$, and consequently, $R\in \Dpo\mathcal{T}^1_3(M)$.
Indeed by the same reasoning that led to \eqref{eq:ext_ltl}, \eqref{eq:L2loc_ext} we may extend the $C^0$-Levi-Civita connection of $g$ to a map
\begin{equation}\label{eq:ext_C0}
    \nabla: \X_{C^0}(M)\times \X_{C^1}(M) \to \X_{C^0}(M), \ \mbox{and}\     
    \nabla: \X_{C^1}(M)\times \X_{C^0}(M) \to \D'^{(1)}\mathcal{T}^1_0(M),
\end{equation}
respectively. Here,
the second equation shows that the term $([\nabla_X,\nabla_Y]Z)(\theta)$ yields an order one distribution,
while $(\nabla_{[X,Y]}Z)(\theta) \in C^0$ since $\nabla$ is a $C^0$-connection.

For $W,X,Y,Z \in \X(M)$ we define $R(W,X,Y,Z)\in \Dpo(M)$ by
\begin{align*}
R(W,X,Y,Z) := &X(g(W,\nabla_Y Z)) - Y(g(W,\nabla_X Z))\\ 
&- g(\nabla_X W,\nabla_Y Z) + g(\nabla_Y W,\nabla_X Z) - g( W,\nabla_{[X,Y]} Z).
\end{align*}
Using \eqref{eq:dist_ext}, \eqref{eq:L2loc_ext} and \eqref{eq:D4D5} it is straightforward to check that cf.\ \cite[Rem.\ 4.5]{LeFMar07}
\[
R(W,Z,X,Y) = g(W,R(X,Y)Z).
\]
Alternatively, one may verify this identity in local coordinates using the fact that there is a well-defined multiplication 
of distributions of first order with $C^1$-functions.

The Ricci tensor of a $C^1$ metric $g$ is given by \eqref{eq:Ric_def}, hence, in particular, is again a distributional tensor field of order $1$. Note  that by \eqref{eq:Ck-ext} this allows us to apply it to $C^1$-vector fields, which will be essential in the context of the genericity condition below. Similarly, since the Riemann tensor, again by \eqref{eq:Ck-ext}, acts on $C^1$-vector fields, we may in particular use $g$-orthonormal frames (which generically are only $C^1$) to calculate the Ricci tensor, which is also vital below. Explicitly, \eqref{eq:Ric_def} holds for $F_i$ a $g$-orthonormal frame and $F^i$ its $g$-dual. 
All these facts can alternatively be  derived directly from the extension of $\nabla$ in \eqref{eq:ext_C0}.

Again, by the same observations, we have that the scalar curvature $S$ is a distribution of order $1$ (this answers a concern from \cite[Sec.\ 3.2]{G20}) and that the standard local formulae hold in $\Dpo$:
\begin{align*}
R^m_{\,ijk} &= \partial_j\Gamma^m_{\,ik} - \partial_k\Gamma^m_{\,ij} + \Gamma^m_{\,js}\Gamma^s_{\,ik} - \Gamma^m_{\,ks}\Gamma^s_{\,ij}, \\
\Ric_{ij} &= R^m_{\,imj},\\
S&=\Ric^i_{\,i}.
\end{align*}
This shows that the local definitions given in \cite[(3.2),(3.3)]{G20} are compatible with the global approach developed here.

\subsection{Regularisation}\label{sec:regularisation}
As in previous works on the generalisation of singularity theorems to metrics of differentiability below $C^2$ \cite{KSSV,hawkingc11,penrosec11,G20}, 
an essential tool in our approach will be suitable regularisations, both of the metrics themselves and of the derived distributional curvature quantities.
A general regularisation scheme based on chart-wise convolution with a mollifier $\rho\in \D(B_1(0))$, $\int \rho = 1$, $\rho\ge 0$ 
(cf.\ \cite[3.2.10]{GKOS}, \cite[Sec.\ 2]{KSSV}, \cite[Sec.\ 3.3]{G20}) is as follows: 
Cover $M$ by a countable and locally finite family of relatively compact chart neighbourhoods 
$(U_i,\psi_i)$ ($i\in \N$). Let $(\zeta_i)_i$ be a subordinate partition of unity with $\mathrm{supp}(\zeta_i)\subseteq U_i$ for all $i$
and choose a family of cut-off functions $\chi_i\in\mathscr{D}(U_i)$ with $\chi_i\equiv 1$ on a
neighbourhood of $\mathrm{supp}(\zeta_i)$. Finally, for $\eps\in (0,1]$ set $\rho_{\eps}(x):=\eps^{-n}\rho\left (\frac{x}{\eps}\right)$.
Then denoting by $f_*$ (respectively $f^*$) push-forward (resp.\ pull-back) of distributions under a diffeomorphism $f$, for any $\mathcal{T} \in \D'\mathcal{T}^r_s(M)$ consider the expression
\begin{equation}\label{eq:M-convolution}
\mathcal{T}\star_M \rho_\eps(x):= \sum\limits_i\chi_i(x)\,\psi_i^*\Big(\big(\psi_{i\,*} (\zeta_i\cdot \mathcal{T})\big)*\rho_\eps\Big)(x).
\end{equation}
Here, $\psi_{i\,*} (\zeta_i\cdot \mathcal{T})$ is viewed as a compactly supported distributional tensor field on $\R^n$, so componentwise convolution
with $\rho_\eps$ yields a smooth field on $\R^n$. The cut-off functions $\chi_i$ secure that $(\eps,x) \mapsto \mathcal{T}\star_M \rho_\eps(x)$
is a smooth map on $(0,1] \times M$. For any compact set $K\comp M$ there is an $\eps_K$ such that for all $\eps<\eps_K$ and
all $x\in K$ \eqref{eq:M-convolution} reduces to a finite sum with all $\chi_i\equiv 1$ (hence to be omitted from the formula), 
namely when $\eps_K$ is less than the distance between the support of $\zeta_i\circ\psi_i^{-1}$ and the boundary of $\psi_i(U_i)$
for all $i$ with $U_i\cap K\neq \emptyset$.

Just as is the case for smoothing via convolution in the local setting, $\mathcal{T}\star_M \rho_\eps$ converges to $\mathcal{T}$ as $\eps\to 0$ in 
$\D'\mathcal{T}^r_s$, and indeed in $C^k_{\mathrm{loc}}$ or $W^{k,p}_{\mathrm{loc}}$ $(p<\infty)$ if $\mathcal{T}$ is contained in these spaces \cite[Prop.\ 3.5]{G20}.

For a given $C^1$ Lorentzian metric $g$ we will in addition require concrete regularisations that are adapted to the causal structure induced
by $g$. This construction goes back to Chrusciel and Grant \cite[Prop.\ 1.2]{CG}, cf.\ also \cite[Prop.\ 2.5]{KSSV}. The version we will use is 
\cite[Lem.\ 4.2, Cor.\ 4.3]{G20}:
\begin{Lemma}\label{Lemma: approximatingmetrics}Let $(M,g)$ be a $C^1$-spacetime. Then for any $\eps>0$ there exist smooth Lorentzian metrics $\check{g}_\eps$, $\hat g_\eps$ on $M$,
time orientable by the same timelike vector field as $g$, and satisfying:
\begin{itemize}
\item[(i)] $\check g_\eps \prec g \prec \hat g_\eps$ for all $\eps$.
\item[(ii)] $\check g_\eps$ and $\hat g_\eps$ converge to $g$ in $C^1_{\mathrm{loc}}$ as $\eps\to 0$.
\item[(iii)] $\check g_\eps - g\star_M \rho_\eps \to 0$ in $C^\infty_{\mathrm{loc}}$ and for any $K\comp M$ there exist $c_K>0$ and
$\eps_0(K)$ such that $\|\check g_\eps -g\star_M \rho_\eps \|_{\infty,K} \le c_K \eps$ and $\|g - g \star_M \rho_\eps\|_{\infty,K} \leq c_K \eps$ for all $\eps < \eps_0$. An analogous statement holds for 
$\hat g_\eps$ as well as for the inverse metrics $g^{-1},(\check g_\varepsilon)^{-1}$ and $(\hat g_\varepsilon)^{-1}$.
\end{itemize}
\end{Lemma}

In this work we will generally use $\check{g}_\eps$, $\hat g_\eps$ to denote the regularisations of Lemma \ref{Lemma: approximatingmetrics}, which are of utmost importance throughout. 

Next we recall the following essential result on convergence of the curvature under regularisations, which is a special case of several so-called stability results, see \cite[Thm.\ 2]{GT87}, \cite[Thm.\ 4.6]{LeFMar07}, \cite[Thm.\ 5.1]{SV09}.

\begin{Lemma}\label{lem:symm} Let $(M,g)$ be a $C^1$-spacetime, and let $g_\eps$ be either $\check{g}_\eps$ or $\hat g_\eps$. Then we have for the Riemann and the Ricci tensor
\begin{equation}
    R[g_\varepsilon]\to R[g]\quad\mbox{and}\quad\Ric[g_\varepsilon]\to\Ric[g]\quad
    \mbox{distributionally.}
\end{equation}
In particular, the Riemann and the Ricci tensor possess the usual symmetries.
\end{Lemma}

We will frequently need the following result from ODE theory (cf.\ \cite[Ch.2, Thm.\ 3.2]{Har02}), which we will mainly use in the context of regularisation given below in Corollary \ref{Corollary: Hartmangeodesicversion}.

\begin{Proposition}\label{Proposition: HartmanODEversion}
Let $f$ and $f_m$, $m \in \mathbb{N}$ be continuous functions on an open set $E \subseteq \R \times \R^n$ with $f_m \to f$ in $\Cloc$. Let $y_m=y_m(t)$ be a solution to the initial value problem
\begin{align*}
    y_m'(t) = f_m(t,y_m(t)), \quad y_m'(t_m) = y_{m_0}
\end{align*}
with data $(t_m,y_{m_0}) \in E$ and maximal interval of existence $(\omega^-_m,\omega^+_m)$. Suppose that the sequence of initial data converges,
\begin{align*}
    (t_m,y_{m_0}) \to (t_0,y_0) \in E.
\end{align*}
Then there exist a solution $y=y(t)$ for the initial value problem
\begin{align*}
    y'(t) = f(t,y(t)), \quad y(t_0)=y_0,
\end{align*}
defined on a maximal interval $(\omega^-,\omega^+)$, and a subsequence $m_l$ with the property that for any $s_1<s_2$ with $(s_1,s_2) \subseteq (\omega^-,\omega^+)$ we have $(s_1,s_2) \subseteq (\omega^-_{m_l},\omega^+_{m_l})$ for all large $l$, and $y_{m_l} \to y$ in 
$\Celoc(\omega^-,\omega^+)$.\footnote{In \cite[Ch.2, Thm.\ 3.2]{Har02}, convergence is only stated in $C^0_{\mathrm{loc}}$, but inserting this into the
equivalent integral equations immediately yields the claim as given here.} In particular,
\begin{align*}
    \limsup_m \omega^-_m \leq \omega^- < \omega^+ \leq \liminf_m \omega^+_m.
\end{align*}
\end{Proposition}

We will mostly use the previous result to conclude that $g_{\varepsilon}$-geodesics, where $g_{\varepsilon} \in \{\check{g}_{\varepsilon},\hat{g}_{\varepsilon}\}$, converge uniformly on compact intervals to a $g$-geodesic:

\begin{Corollary}\label{Corollary: Hartmangeodesicversion}
Let $(M,g)$ be a $C^1$-spacetime and let $g_m = \check{g}_{\varepsilon_m}$ or $g_m = \hat{g}_{\varepsilon_m}$ with $\varepsilon_m \to 0$. Let $\gamma_m:(a_m,b_m) \to M$ be inextendible $g_m$-geodesics. Suppose that $\gamma_m(t_0) \to p \in M$ and $\dot{\gamma}_m(t_0) \to v\in T_pM$, where $t_0 \in \bigcap (a_m,b_m)$. Then there exist an inextendible $g$-geodesic $\gamma:(a,b) \to M$ with $\gamma(t_0)=p$, $\dot \gamma(t_0)=v$ and a subsequence $m_l$ such that whenever $(s_1,s_2) \subseteq (a,b)$, we have $(s_1,s_2) \subseteq (a_{m_l},b_{m_l})$ for all large $l$ and $\gamma_{m_l}$ converges in $C^2_{\mathrm{loc}}(a,b)$ to $\gamma$.
Moreover,
\begin{align*}
    \limsup_m a_m \leq a < b \leq \liminf_m b_m.
\end{align*}
\end{Corollary}
\begin{proof}
As this is a local statement, we may assume $M = \R^n$. We rewrite the geodesic equations for $g_m$ as a first order initial value problem (with $\Gamma_m$ the Christoffel symbols of $g_m$) and $1\le i \le n$:
\begin{align*}
    &\dot{\gamma}_m^i(t) = \eta_{m}^i(t),\\
    &\dot{\eta}_m^i = - (\Gamma_m)^i_{jk}(\gamma_m(t))\eta_m^j(t) \eta_m^k(t),\\
    &\gamma_m^i(t_0) = p_m^i,\\
    & \eta_m^i(t_0)  = v_m^i.
\end{align*}
We define the continuous functions $f_m:\R^{2n} \to \R$, $f_m(x,y):=(y,-(\Gamma_m)^i_{jk}(x) y^j y^k)$, and similarly $f:\R^{2n} \to \R$, $f(x,y):=(y,-\Gamma_{jk}(x) y^j y^k)$. Then $f_m \to f$ uniformly on compact subsets of $\R^{2n}$, and by assumption $(p_m,v_m) \to (p,v) \in \R^{2n}$. The existence of an inextendible $g$-geodesic sublimit $\gamma$ of the sequence $\gamma_m$ now follows from Proposition \ref{Proposition: HartmanODEversion}, and the $C^2_{\mathrm{loc}}$-convergence follows suit.
\end{proof}

In the case of only future or past inextendibility, an obvious analogue of
Corollary \ref{Corollary: Hartmangeodesicversion} holds.
\begin{remark}\label{remark:parallel_transport} (Parallel transport)

\noindent
Under the assumptions of Corollary \ref{Corollary: Hartmangeodesicversion}, let $w_m\in T_{\gamma_m(t_0)}M$, $w_m \to w \in T_pM$. 
Then since the ODEs governing parallel transport are linear, there are unique smooth vector fields $W_m$ along $\gamma_m$ that are global solutions 
of the corresponding initial value problem, and by the standard results on continuous dependence of solutions on the right hand side and the data, they converge in $C^1_{\mathrm{loc}}$ to the unique $C^1$-vector field $W$ along $\gamma$ resulting from $g$-parallel transport of $w$ along $\gamma$.
\end{remark}

\subsection{Energy Conditions}
\label{sec:EC}
Let us briefly recall the timelike and null energy conditions for $C^1$-spacetimes. They are formulated in a way so that they imply corresponding conditions for the curvatures of approximating metrics. If the metric is $C^{1,1}$ or better, then the energy conditions presented here are equivalent to the usual pointwise (or pointwise a.e.) energy conditions. We follow \cite[Sec.\ 3-5]{G20}, to which we also refer for proofs.\\
We start with the timelike energy condition, whose distributional generalisation is straightforward, cf.\ \cite[Def.\ 3.3]{G20}.

\begin{definition} (Distributional timelike energy condition)
\label{definition: distributionaltimelikeenergycondition}

\noindent Let $(M,g)$ be a $C^1$-spacetime. We say that $(M,g)$ satisfies the \textit{distributional timelike energy condition} if for any timelike
$X \in \mathfrak{X}(M)$, $\Ric(X,X) \in \mathcal{D}^{'(1)}(M)$ is a nonnegative scalar distribution on $M$.
\end{definition}

Recall that a scalar distribution $u\in\D'(M)$ is nonnegative, $u\geq 0$, if $u(\omega)\equiv\langle u,\omega\rangle\geq 0$ for all nonnegative test densities $\omega\in\Gamma_c(M,\Vol(M))$. A nonnegative distribution is always a measure \cite[Thm.\ 2.1.7]{Hoe90} and hence a distribution of order $0$. Moreover, non-negativity is stable with respect to regularisation \cite[Thm.\ 2.1.9]{Hoe90}. Finally, for $u,v\in\D'(M)$ we write $u\geq v$ if $u-v\geq 0$. 

\begin{remark} We note that an equivalent formulation of the distributional timelike energy condition consists in imposing that, 
for any $U\subseteq M$ open and any $X\in \X(U)$ timelike, $\Ric(X,X)\ge 0$ in $\mathcal{D}^{'(1)}(U)$. Indeed, given a timelike $X\in \X(U)$, let $V\comp U$ be open and choose a cut-off function $\chi\in \D(U)$ that equals $1$ on a neighbourhood of $\overline V$. By time-orientability of $M$ there exists some timelike $T\in \X(M)$ of the same time-orientation as $X$, so $Y:=\chi X +(1-\chi)T$ is a global timelike vector field on $M$. Then for any non-negative $\omega\in \Gamma_c(M,\Vol(M))$ with support in $V$ we have $0\le \langle \Ric(Y,Y),\omega\rangle = \langle \Ric(X,X),\omega\rangle$, and by the sheaf property of $\D^{'(1)}$, we obtain $\Ric(X,X)\ge 0$ in $\D^{'(1)}(U)$.
\end{remark}

The usefulness of Definition \ref{definition: distributionaltimelikeenergycondition} is highlighted by its implications for the curvatures of 
approximating metrics, cf.\ \cite[Lem.\ 4.1, Lem.\ 4.6]{G20}.

\begin{Lemma}\label{Lemma: timelikeenergyconditionapproxmetrics1}
Let $(M,g)$ be a $C^1$-spacetime satisfying the distributional timelike energy condition and let $K$ be compact in $M$. Then
\begin{align*}
    &\forall C>0 \ \forall \delta > 0 \ \forall \kappa < 0 \ \exists \varepsilon_0 > 0 \ \forall  \varepsilon < \varepsilon_0 \ \forall  X \in TM|_K\\
    &\text{with } g(X,X) \leq \kappa \text{ and } \|X\|_h \leq C: \ \Ric[\check{g}_{\varepsilon}](X,X) > - \delta.
\end{align*}
\end{Lemma}

The following is an immediate consequence of the previous result.

\begin{Corollary}\label{Corollary: timelikeenergyconditionapproxmetrics2}
Let $(M,g)$ be a $C^1$-spacetime satisfying the distributional timelike energy condition and let $K$ be compact in $M$. Then
\begin{align*}
    \forall \delta > 0 \ \exists \varepsilon_0 > 0 \ \forall \varepsilon < \varepsilon_0 \ \forall X \in TM|_K \text{ with } \check{g}_{\varepsilon}(X,X) = -1:\ \Ric[\check{g}_{\varepsilon}](X,X) > - \delta.
\end{align*}
\end{Corollary}

Let us now turn to the proper distributional formulation of the null energy condition. The naive approach of defining it in analogy to Definition \ref{definition: distributionaltimelikeenergycondition} does not yield analogues of Lemma \ref{Lemma: timelikeenergyconditionapproxmetrics1} and Corollary \ref{Corollary: timelikeenergyconditionapproxmetrics2} because vector fields that are $\check{g}_{\varepsilon}$-null are only almost $g$-null. This necessitates a definition of the following form. 

\begin{definition}(Distributional null energy condition, \cite[Def.\ 5.1]{G20})
\label{definition: distributionalnullenergycondition}

\noindent Let $(M,g)$ be a $C^1$-spacetime. We say $(M,g)$ satisfies the \textit{distributional null energy condition} if for any compact set $K$ and any $\delta > 0$ there exists $\varepsilon = \varepsilon(\delta,K) > 0$ such that $\Ric(X,X) > - \delta$ (as distributions) for any local smooth vector field 
$X \in \mathfrak{X}(U)$ ($U \subseteq K$ open) with $\|X\|_h = 1$ and $|g(X,X)| < \varepsilon$ on $U$.
\end{definition}

The above requirement $|g(X,X)| < \varepsilon$ may equivalently be replaced by $\|X - N\|_h < \varepsilon$, where $N$ is a $C^1$ $g$-null vector field on $U$, cf.\ \cite[Lem.\ 5.2]{G20}. Yet another useful equivalent reformulation is the following rescaled version, cf.\ \cite[Def.\ 5.3]{G20}.

\begin{Lemma}\label{Lemma: distribnullequivalent} A $C^1$-spacetime $(M,g)$ satisfies the distributional null energy condition if and only if for any compact $K\subseteq M$, any $c_1,c_2 > 0$ and any $\delta > 0$ there is $\varepsilon=\varepsilon(\delta,K,c_1,c_2) > 0$ such that $\Ric(X,X) > - \delta$ (as distributions) for any local smooth vector field $X \in \mathfrak{X}(U)$, $U \subseteq K$, with $0 < c_1 \leq \|X\|_h \leq c_2$ and $|g(X,X)| < \varepsilon$ on $U$.
\end{Lemma}

The following analogue of Lemma \ref{Lemma: timelikeenergyconditionapproxmetrics1} and Corollary \ref{Corollary: timelikeenergyconditionapproxmetrics2} shows that the above definition of the null energy condition is the correct one in the $C^1$-case (see \cite[Lem.\ 5.5]{G20}).

\begin{Lemma}\label{Lemma: nullenergyconditionapproxmetrics}
Let $(M,g)$ be a $C^1$-spacetime satisfying the distributional null energy condition. Let $K\subseteq M$ be compact and let $c_1,c_2 > 0$. Then for all $\delta > 0$ there is $\varepsilon_0 = \varepsilon_0(\delta,K,c_1,c_2) > 0$ such that
\begin{align*}
    \forall \varepsilon < \varepsilon_0 \ \forall X \in TM|_K \text{ with } 0 < c_1 \leq \|X\|_h \leq c_2 \text{ and } \check{g}_{\varepsilon}(X,X) = 0:\ \Ric[\check{g}_{\varepsilon}](X,X) > - \delta.
\end{align*}
    
\end{Lemma}

The above distributional versions of the timelike and null energy conditions were successfully used in \cite{G20} to prove the singularity theorems of Hawking and of Penrose for $C^1$-spacetimes. For our $C^1$-proof of the Hawking-Penrose singularity theorem, we need a distributional version of the genericity condition along geodesics. 

\begin{definition}(Distributional genericity condition)
\label{definition: strongdistributionalgenericity}

\noindent Let $(M,g)$ be a $C^1$-spacetime and let $\gamma:I \to M$ be a causal geodesic. We say that the \textit{distributional genericity condition} holds at $\gamma(t_0)$, $t_0 \in I$, if there is a neighbourhood $U$ of $\gamma(t_0)$, $C^1$-vector fields $X,V$ on $U$ with the property that $(X \circ \gamma)(t) = \dot{\gamma}(t)$, $(V \circ \gamma)(t) \perp \dot{\gamma}(t)$ for all $t \in I$ with $\gamma(t) \in U$, and there exist $c > 0$ and $\delta > 0$ such that for all $C^1$-vector fields $\tilde{X}$, $\tilde{V}$ on $U$ with $\|X - \tilde{X}\|_h < \delta$ and $\|V - \tilde{V}\|_h < \delta$ we have
\begin{equation}\label{eq:genericity}
    g(R(\tilde{X},\tilde{V})\tilde{V},\tilde{X}) > c \quad \text{in } \mathcal{D}^{'(1)}(U).
\end{equation}
\end{definition}

This condition extends the one used in the $C^{1,1}$-case in \cite[Def.\ 2.2]{GGKS} as we shall discuss in Remark \ref{remark: distribgenericityequivalentinC11}, below. Again it allows us to derive an approximate genericity condition for regularisations. Note, however, that it will turn out to be essential to use $C^1$-vector fields (and their approximations), which can be inserted into  $R\in{\D'}^{(1)}{\mathcal T}^1_3(M)$, by \eqref{eq:Ck-ext}.

Also the distributional genericity condition appropriately restricts to smaller sets, a technical point which we clarify first.
\begin{remark}\label{rem:localizegen}
In Definition \ref{definition: strongdistributionalgenericity} we may shrink the neighbourhood $U$ while retaining the constants $c$ and $\delta$. More precisely, let $U$, $X$, $V$, $c$, and $\delta$ be as in the definition and let $W\comp U$, then \eqref{eq:genericity} holds for all $C^1$-vector fields $\tilde X_W$, $\tilde V_W$ on $W$ which satisfy $\|X|_W-\tilde X_W\|_h<\delta$ and $\|V|_W-\tilde V_W\|_h<\delta$. Indeed, by the sheaf property of distributions, it suffices to establish that \eqref{eq:genericity} then holds locally around each point $x\in W$. To establish this let $\chi\in\D(W)$ be a plateau function that equals $1$ in a neighbourhood $W'$ of $x$ and set $\tilde X=X+\chi(\tilde X_W- X)$ and likewise for $\tilde V$. Then $\tilde X$ and $\tilde V$ are $C^1$-vector fields on $U$  with $\|X - \tilde{X}\|_h < \delta$ and $\|V - \tilde{V}\|_h < \delta$ 
 \eqref{eq:genericity}. So we obtain on $W'$ 
\begin{equation}
  c< g(R(\tilde{X},\tilde{V})\tilde{V},\tilde{X}) = g(R(\tilde{X}_W,\tilde{V}_W)\tilde{V}_W,\tilde{X}_W).
\end{equation}
\end{remark}

\begin{Lemma}
\label{Lemma: genericityfriedrichs}
Let $(M,g)$ be a $C^1$-spacetime and let $g_{\varepsilon}$ be either
$\check{g}_{\varepsilon}$ or $\hat{g}_{\varepsilon}$. Suppose $\gamma$ is a
$g$-causal $g$-geodesic and assume that the distributional genericity condition holds at $\gamma(t_0)$ (with neighbourhood $U$, vector fields $X,V$ and constants $c,\delta$). Suppose that $X_{\varepsilon}$, $V_{\varepsilon}$ are $C^1$-vector fields on $U$ with $X_{\varepsilon}, V_{\varepsilon} \to X,V$ in
$C^1_{{\mathrm{loc}}}(U)$. Then for any compact subset $K\comp U$
there is $\tilde{\varepsilon} > 0$ such that for all $\varepsilon < \tilde{\varepsilon}$:
\begin{align*}
    g_{\varepsilon}(R[g_{\varepsilon}](X_{\varepsilon},V_{\varepsilon})V_{\varepsilon},X_{\varepsilon}) > \frac{c}{2}
\end{align*}
on $K$.
\end{Lemma}
\begin{proof}
We may take $U$ to be relatively compact. By assumption, for all $C^1$-vector fields $\tilde{X},\tilde{V}$ on $U$ with $\|X - \tilde{X}\|_h < \delta$ and $\|V - \tilde{V}\|_h < \delta$, we have
\begin{align*}
     g(R(\tilde{X},\tilde{V})\tilde{V},\tilde{X}) > c \quad \text{in } \mathcal{D}^{'(1)}(U).
\end{align*}
Since $X_{\varepsilon} \to X$ and $V_{\varepsilon} \to V$ in $C^1_{{\mathrm loc}}$, there is $\varepsilon_0 > 0$ such that for all $\varepsilon < \varepsilon_0$, $\|X_{\varepsilon} - X\|_h < \delta$ and $\|V_{\varepsilon} - V\|_h < \delta$. Let from now on $\varepsilon < \varepsilon_0$. Hence, the assumption gives
\begin{align*}
    g(R(X_{\varepsilon},V_{\varepsilon})V_{\varepsilon},X_{\varepsilon}) > c \quad \text{in } \mathcal{D}^{'(1)}(U).
\end{align*}
Since $\star_M$-convolution with a non-negative mollifier $\rho_{\varepsilon}$
respects positivity, we have
\begin{align*}
    g(R(X_{\varepsilon},V_{\varepsilon})V_{\varepsilon},X_{\varepsilon}) \star_M \rho_{\varepsilon} > c.
\end{align*}
We claim that
\begin{align*}
    g(R(X_{\varepsilon},V_{\varepsilon})V_{\varepsilon},X_{\varepsilon}) \star_M \rho_{\varepsilon} - (g \star_M \rho_{\varepsilon}) (R[g\star_M \rho_{\varepsilon}](X_{\varepsilon},V_{\varepsilon})V_{\varepsilon},X_{\varepsilon}) \to 0 \quad \text{in } C^0_{\mathrm{loc}}(U).
\end{align*}
This is seen by first showing that
\begin{align*}
    g(R(X_{\varepsilon},V_{\varepsilon})V_{\varepsilon},X_{\varepsilon}) \star_M \rho_{\varepsilon} - g((R[g] \star_M \rho_{\varepsilon})(X_{\varepsilon},V_{\varepsilon})V_{\varepsilon},X_{\varepsilon}) \to 0 \quad \text{in } C^0_{\mathrm{loc}}(U).
\end{align*}
The latter locally corresponds to the question whether given a net $h_{\varepsilon}$ of $C^1$-functions
that converges in $C^1_{\mathrm{loc}}$ to $h \in C^1$ and a first order distribution $T$ we can show
$(h_{\varepsilon}T) *\rho_{\varepsilon} - h_{\varepsilon} (T * \rho_{\varepsilon}) \to 0$ in
$C^0_{\mathrm{loc}}$. This is shown by arguing along the lines of \cite[Lem.\ 4.8]{G20} and using that $T$ can locally be
written as a first order distributional derivative of a $C^0$-function.
By a similar argument
(noting the fast convergence speed of $g \star_M \rho_{\varepsilon} \to g$), the right hand side
above is $C^0_{\mathrm{loc}}$-equivalent (i.e.\ the difference goes to $0$ in $C^0_{\mathrm{loc}}$) to
\begin{align*}
    (g \star_M \rho_{\varepsilon})(R[g] \star_M \rho_{\varepsilon}(X_{\varepsilon},V_{\varepsilon})V_{\varepsilon},X_{\varepsilon}).
\end{align*}
This is $C^0_{\mathrm{loc}}$-equivalent to
\begin{align*}
    (g \star_M \rho_{\varepsilon})(R[g \star_M \rho_{\varepsilon}](X_{\varepsilon},V_{\varepsilon})V_{\varepsilon},X_{\varepsilon})
\end{align*}
by the same arguments as in \cite[Lem.\ 4.6]{G20} and uniform boundedness of $X_\eps$, $V_{\varepsilon}$. 

Next, we show that
\begin{align*}
    (g \star_M \rho_{\varepsilon})(R[g \star_M \rho_{\varepsilon}](X_{\varepsilon},V_{\varepsilon})V_{\varepsilon},X_{\varepsilon}) - g_{\varepsilon}(R[g_{\varepsilon}](X_{\varepsilon},V_{\varepsilon})V_{\varepsilon},X_{\varepsilon}) \to 0
\end{align*}
in $C^0_{\mathrm{loc}}(U)$. Upon writing everything out in coordinates, we see that only terms of the form
\begin{align*}
    (g \star_M \rho_{\varepsilon}) \partial \Gamma[g \star_M \rho_{\varepsilon}] X_{\varepsilon} V_{\varepsilon} V_{\varepsilon} X_{\varepsilon} - g_{\varepsilon} \partial \Gamma[g_{\varepsilon}] X_{\varepsilon} V_{\varepsilon} V_{\varepsilon} X_{\varepsilon}
\end{align*}
are of interest, because for any of the other terms in the difference, both summands converge in $C^0_{\mathrm{loc}}(U)$ to the respective expressions for $g$. Since the $X_{\varepsilon}$, $V_{\varepsilon}$ are uniformly bounded, we may omit them in our estimates without loss of information. But for terms as above, we only need to consider those expressions where we have second derivatives of the metrics, since $\partial_j (g \star_M \rho_{\varepsilon})^{mk} \to \partial_j g^{mk}$ in $C^0_{\mathrm{loc}}(U)$ and similarly for $\partial_j (g_{\varepsilon})^{mk}$. Hence, we work in a chart, fix a compact set $K$ and are left with estimating the following expression on $K$:
\begin{align*}
    &|(g \star_M \rho_{\varepsilon})_{bc} (g \star_M \rho_{\varepsilon})^{mk} \partial_s \partial_r (g \star_M \rho_{\varepsilon})_{ij} - (g_{\varepsilon})_{bc} (g_{\varepsilon})^{mk} \partial_s \partial_r (g_{\varepsilon})_{ij}|\\
    &\leq |(g \star_M \rho_{\varepsilon})_{bc} (g \star_M \rho_{\varepsilon})^{mk}| \cdot |\partial_s \partial_r (g \star_M \rho_{\varepsilon})_{ij} - \partial_s \partial_r (g_{\varepsilon})_{ij}| + II.
\end{align*}
The first term goes to $0$ uniformly on $K$ by Lemma \ref{Lemma: approximatingmetrics}(iii). We continue estimating the remaining term $II$:
\begin{align*}
    II &= |(g \star_M \rho_{\varepsilon})_{bc} (g \star_M \rho_{\varepsilon})^{mk} \partial_s \partial_r (g_{\varepsilon})_{ij} - (g_{\varepsilon})_{bc} (g_{\varepsilon})^{mk} \partial_s \partial_r (g_{\varepsilon})_{ij}|\\
    &\leq |\partial_s \partial_r (g_{\varepsilon})_{ij}| \cdot |(g \star_M \rho_{\varepsilon})_{bc} (g \star_M \rho_{\varepsilon})^{mk}  - (g_{\varepsilon})_{bc} (g_{\varepsilon})^{mk}|.
\end{align*}
To show that this goes to $0$ uniformly on $K$, it suffices to show that the second factor can be
estimated by a constant times $\varepsilon$ by the same argument as in the proof of \cite[Lem.\ 4.5]{G20}. Since $\|g_{\varepsilon} - (g \star_M \rho_{\varepsilon})\|_{\infty,K}$ and $\|g_{\varepsilon}^{-1} - (g \star_M \rho_{\varepsilon})^{-1}\|_{\infty,K}$ can both be estimated by
a constant times $\varepsilon$ again by Lemma \ref{Lemma: approximatingmetrics}(iii), this is easily seen to hold.

Hence, there is $\varepsilon_1 > 0$ such that for all $\varepsilon < \min(\varepsilon_0,\varepsilon_1)$,
\begin{align*}
    g_{\varepsilon}(R[g_{\varepsilon}](X_{\varepsilon},V_{\varepsilon})V_{\varepsilon},X_{\varepsilon}) > \frac{c}{2} \quad \text{on } K.
\end{align*}
\end{proof}

Next we derive from the distributional genericity condition an estimate on the tidal force operator for the regularised metrics. It is this result which allows us to enter into the Riccati comparison techniques in Section \ref{sec:max_geods}.

\begin{Lemma}
\label{Lemma: tidalforcematrixestimateapprox}
    Let the assumptions of Lemma \ref{Lemma: genericityfriedrichs} hold and 
    let $\gamma_\varepsilon$ be $g_{\varepsilon}$-causal $g_{\varepsilon}$-geodesics, whose
    $g_{\varepsilon}$-causal character is the same as the $g$-causal character of $\gamma$
    and which converge in $C^1_{\mathrm{loc}}$ to $\gamma$. 
    Then there are $\tilde c>0$, $r>0$, $\eps_0>0$ and a neighbourhood $U$ of $\gamma(t_0)$
    such that for any $0<\eps<\eps_0$ there
    exist vector fields $E_i^\eps$ on $U$ such that $E_i^\eps \circ \gamma_\eps$ constitutes a
    $g_\eps$-orthonormal frame along $\gamma_\eps$ and such that $E_i^\eps \to E_i$ in
    $C^1_{\mathrm{loc}}(U)$, where $E_i\circ \gamma$ is an orthonormal frame along $\gamma$.
    Further, there  exists some
    $C=C(\eps)>0$ such that along $\gamma_\eps$ the $\eps$-tidal force operator
    $[R_\eps](t):= [R_\eps(.,\dot{\gamma}_\eps(t))\dot{\gamma}_\eps(t)]:[\dot{\gamma}_\eps(t)]^\perp \to
    [\dot{\gamma}_\eps(t)]^\perp$ fulfills  
	\begin{equation}
		[R_\eps](t) > \mathrm{diag}(\tilde c,-C,\ldots,-C) \text{ on } [t_0-r,t_0+r]
	\end{equation}
	in terms of the frame $E_i^\eps$ and where we have used the shorthand $R_\varepsilon$ for $R[g_\varepsilon]$.
\end{Lemma}

\begin{proof}
    
    Let w.l.o.g.\ $t_0=0$. We deal with the case of $\gamma$ being null or timelike 
    simultaneously and construct a $g$-orthonormal frame for $[\dot{\gamma}]^\perp$ along $\gamma$. First note
    that, if $\gamma$ is timelike, $V\circ\gamma$
    is nowhere proportional to $\dot{\gamma}$, but the same may not be true in the null case.
    First we will show that we can replace $V$ by a vector field whose restriction to $\gamma$ is parallel  in $[\dot{\gamma}]^\perp$ and still retain
    \eqref{eq:genericity}.
    
    There must exist some $t$ such that $V\circ \gamma(t)\not \in \text{span}(\dot{\gamma}(t))$. Indeed, suppose 
    to the contrary that $V(\gamma(t)) = f(t)\cdot X(\gamma(t))$ for some real-valued $C^1$ function $f$ and all $t$. 
    Then extending both $X\circ \gamma$ and $V\circ \gamma$ in a 
    cylindrically constant fashion as given by Lemma \ref{lem:A1}, we obtain vector fields $\tilde V, \tilde X$ on a neighbourhood $W$
    of $\gamma(0)$ that are proportional and can be made arbitrarily close to $V$ resp.\ $X$ in $C^0$ when shrinking $W$
    since they restrict to the same vector fields on $\gamma$. According to Definition \ref{definition: strongdistributionalgenericity}
    they therefore satisfy \eqref{eq:genericity} close to $\gamma(0)$. However, by the symmetry properties of $R$ (cf.\ Lemma \ref{lem:symm}) and the proportionality of $\tilde V$ and $\tilde X$ we must have $R(\tilde X,\tilde V)\tilde V =0$ in $\D'(W)$, a contradiction.

    Hence there exists some $v:=V(t_1)$ with $v\in [\dot{\gamma}(t_1)]^\perp$. Defining $V_1$ as the
    parallel translate of $v$ along $\gamma$ we obtain $V_1\in [\dot{\gamma}]^\perp$ along all of $\gamma$. 
    By the same continuity argument as above, \eqref{eq:genericity}  
    then still holds for the vector fields $V_1$ and $X$ (note that a $t_1$ as above must exist arbitrarily close to $0$). 
    In the timelike case we may assume that $\gamma$ is
    parametrised to unit speed, and we set $e_n:=\dot{\gamma}(0)$. If $\gamma$ is null we can write $\dot{\gamma}(0)=e_{n-1} + e_n$
    for orthonormal vectors $e_{n-1}, e_n$ with $e_n$ timelike.

    In the null case, extend $e_1:=V_1(0)\in [\dot{\gamma}(0)]^\perp$ to a $g$-orthonormal basis $\{e_i\}_{i=1}^n$ (with $e_{n-1}, e_n$ as above)
    and define $E_i$ to be the $g$-parallel translate of $e_i$ along $\gamma$ (so that $E_1=V_1$ and $X=E_{n-1}+E_n$ along $\gamma$).
    For $\gamma$ timelike we also construct an orthonormal
    frame $E_1=V_1, E_2,\dots, E_n=X$ along $\gamma$ via parallel transport.
    As $\gamma_\eps\to \gamma$ in $C^1_{\mathrm{loc}}$, we can find a $g_\eps$ orthonormal basis $\{e^\eps_i\}$ at
    $\gamma_\eps(0)$ with $e^\eps_i\to e_i$ in $TM$, and such that $\dot{\gamma}_\eps(0) = e^\eps_n$ if $\gamma_\eps$ is $g_\eps$-timelike, while
    $\dot{\gamma}_\eps(0) = e^\eps_{n-1}+e^\eps_n$ with $e^\eps_n$ timelike if $\gamma_\eps$ is $g_\eps$-null.
    For every $\eps$ let $E_i^\eps$ be the $g_\eps$-parallel
    translate of $e^\eps_i$ along $\gamma_\eps$. By Remark \ref{remark:parallel_transport}, $E_i^\eps \to E_i$ in $C^1_{\mathrm{loc}}$ 
    (i.e.\ in $T(TM)$, uniformly on compact time intervals) and we can use Lemma \ref{lem:A2} to extend $E_i^\eps$ to vector fields on $U$,
    converging in $C^1_{\mathrm{loc}}(U)$ to the extensions of $E_i$ to $U$.
    
    As we have now found vector fields $E^\eps_i$ that restrict to a $g_\eps$-frame along $\gamma_\eps$ and converge to the $g$-frame $E_i$ along $\gamma$, 
    we are in a position to apply Lemma \ref{Lemma: genericityfriedrichs} and hence estimate the $g_\eps$-tidal force operator along $\gamma_\eps$.
    
    For the remainder of the proof we roughly follow the layout of the proof of \cite[Prop. 3.6]{GGKS}. There, for a
    $C^{1,1}$-metric $g$ the statement is first shown for the $g$-tidal force operator along $\gamma$ and then carried over for small $\eps$ by an approximation argument. However, in the $C^1$-setting this is no longer possible, as $R$ is now merely distributional and for a matrix with distributional entries eigenvalues are not defined. Also, a direct approximation procedure of that form is no longer feasible. 
    Instead we use Lemma \ref{Lemma: genericityfriedrichs}, and argue
    separately for any $\epsilon$, making the procedure (and hence the constant $C$) dependent on $\epsilon$\footnote{This is, however, sufficient in our proof of the main result.}. More precisely:
	
	Let $c$ be the constant given by the genericity condition and let $\tilde c:=\frac{c}{2}$ and $0<c_1<\tilde c$.
	Set $R^\epsilon_{ij}:= \langle R_\epsilon(E^\epsilon_i, X_\epsilon) X_\epsilon, E^\epsilon_j
	\rangle$, $i,j=1,\dots,d$, where $X_\eps = E_n^\eps$ in the timelike case and $X_\eps = E_{n-1}^\eps + E_n^\eps$ in the null case.
	As in the proof of $(3.5)$ in \cite[Lemma 3.7]{GGKS}, we choose $C(\epsilon) >0$ such that 	$\lambda^\eps_\text{min}(x) \geq \frac{\|(R^\eps_{1j})_j\|_e}{\tilde c-c_1}$ for all $x\in U$,
	where $\lambda^\eps_\text{min}(x)$ is the smallest eigenvalue 
	of $(R^\eps_{ij})_{i,j=2}^d+C(\epsilon) \text{Id}_{d-1}$ at $x\in U$. By Lemma \ref{Lemma: genericityfriedrichs} there exists some $\tilde \eps>0$
	such that $R^\epsilon_{11}>\tilde c$ for all $0<\eps<\tilde \eps$ on $U$ (shrinking $U$ once more if necessary). For these $\eps$ it then follows 
	that $R^\eps_{ij}-\text{diag}(c_1, - C(\epsilon), \ldots-C(\epsilon))$ is positive definite on $U$. Then picking 
	$r>0$ such that $\gamma_\epsilon([-r,r]) \subseteq U$ for all $\epsilon$ small, evaluation along $\gamma_\eps$ gives the claim.
\end{proof}

\begin{remark}\label{remark: distribgenericityequivalentinC11} \ 
Note that while Definition \ref{definition: strongdistributionalgenericity} may at first sight appear to be
stronger than a straightforward generalisation of the genericity condition given for $C^{1,1}$-metrics in
\cite{GGKS}, due to the local boundedness of $R$ in that regularity it is actually equivalent:

For $C^{1,1}$-metrics Definition \ref{definition: strongdistributionalgenericity} clearly implies 
Definition 2.2 of \cite{GGKS}. For the reverse implication, let $X, V$ be continuous vector fields
in a neighbourhood of $\gamma(t_0)$ such that $X(\gamma(t))=\dot{\gamma}(t)$ and $V(\gamma(t))\in (\dot{\gamma}(t))^\perp$ for all
$t$ and such that $\langle R(V,X)X,V\rangle>c>0$. Arguing exactly as in the proof of Lemma \ref{Lemma: tidalforcematrixestimateapprox}
we can then use parallel transport and cylindrically constant extension to find $C^1$-vector fields 
that are uniformly close to $X$ and $V$.  It therefore suffices to show that the genericity condition
from \cite{GGKS} is stable under locally uniform perturbation. Thus let
$\tilde V, \tilde X $ be $C^1$ vector fields that are uniformly close to $V, X$. Then one can locally in a coordinate system estimate
$|\tilde V_i - V_i| \le \delta $, $|\tilde X_i - X_i| \le \delta $ for some small $\delta$. Consequently, 
\begin{align*}
    \langle R(\tilde V, \tilde X)\tilde X,\tilde V \rangle =  g_{ij}R^i_{klm} \tilde V^l \tilde
    X^m \tilde X^k \tilde V^j \geq  g_{ij}R^i_{klm}V^l X^m X^k V^j+ \mathcal{O} (\delta)
\end{align*}
where the last term collects all $\delta$-contributions and takes into account the local boundedness of $R$. Consequently, choosing $\delta$ small enough (and appropriately shrinking the neighbourhood of $\gamma(t_0)$) we can secure that the sum still remains $>c/2$.
\end{remark}

\section{Branching}\label{sec:branching}

In metric geometry, where geodesics are defined as (local) minimisers of the length functional, local uniqueness of geodesics
is expressed in the form of a \emph{non-branching} condition. Here, a branch point is defined as an element of a minimiser
at which the curve splits into two minimisers that on some positive parameter interval do not have another point in common, cf., e.g., \cite{shio, vil09}. Similarly, in the synthetic Lorentzian setting \cite{KS18}, the role of causal geodesics is
taken on by maximising causal curves, and non-branching is formulated analogously. In both cases, lower synthetic sectional 
curvature bounds (formulated via triangle comparison in constant curvature model spaces) imply non-branching \cite{shio,KS18}. 
Synthetic Ricci curvature bounds in metric measure spaces (curvature-dimension
conditions), on the other hand, do not imply non-branching of geodesics \cite{ohta14}. Similarly, in recent work on synthetic
Ricci curvature bounds in Lorentzian pre-length spaces \cite{CV20}, a timelike non-branching condition is required in addition to
a timelike-curvature dimension condition to obtain a version of the Hawking singularity theorem.

In the $C^1$-setting we are concerned with in this paper, the coincidence between causal local maximisers and geodesics that 
is familiar from smooth Lorentzian geometry ceases to hold (cf.\ Example \ref{ex:branching} below), although it is still 
true that causal maximisers are geodesics (Lemma \ref{Lemma: maxunbrokengeodesic}). In addition, one generically
cannot expect unique local solvability of the geodesic initial value problem. Nevertheless, on physical grounds, it seems reasonable
to assign a privileged role to causal geodesics that \emph{are} locally maximising. We therefore introduce a (weak notion of a) non-branching 
condition that is intended to preclude locally maximising geodesics from branching (where we do not require the second
branch to be maximising as well):

\begin{definition}\label{def:branching} (Non-branching conditions) 

\noindent    Let $(M,g)$ be a $C^1$-spacetime. A geodesic $\gamma:[a,b]\to M$ branches at $t_0\in (a,b)$ if there 
    exist $\eps>0$ and some geodesic $\sigma$ with $\gamma|_{[t_0-\eps, t_0]}\subseteq \sigma$, but
    $\gamma|_{(t_0,t_0+\eps)}\cap \sigma= \emptyset$.
    
    \noindent $(M,g)$ is called maximally causally (resp.\ timelike, resp.\ null) non-branching (MCNB, MTNB, MNNB), if no maximal causal (resp.\ timelike, resp.\ null) 
    geodesic branches in the above sense.
\end{definition}

\begin{exm}\label{ex:branching}
A $C^1$-spacetime in which maximal causal branching occurs can be constructed from the second Riemannian example given in \cite{HW51}. There, a $C^{1,\alpha}$-Riemannian metric ($\alpha<1$) in the $(u,v)$-plane is given with non-unique solutions to the geodesic initial value problem starting at $\{v=0\}$. Further it is shown that the geodesic boundary value problem for geodesics $\gamma$ starting at the surface $\{v=0\}$ is uniquely solvable and hence such geodesics are at least initially minimizing: If they were not, by properties of any $C^1$-Riemannian manifold as a locally compact length space, there would exist a minimiser $\sigma$ between two points on $\gamma$. However in $C^1$-Riemannian manifolds minimisers are geodesics and by uniqueness (of the boundary value problem) $\gamma=\sigma$ (cf.\ \cite{SS18} for proofs and references).

In \cite{HMSSS} it is shown that causal geodesics in static spacetimes are maximising if and only if their Riemannian parts are minimizing and so one can easily construct both maximising timelike and null geodesics that branch.

\end{exm}

\begin{Lemma}
\label{Lemma: nonbranchingmeettangentiallyequal}
    If the $C^1$-spacetime $(M,g)$ is MCNB (resp.\ MTNB, MNNB), then any two maximal causal (resp.\ timelike, null) geodesics that meet tangentially at an interior point must coincide on their maximal domain of definition.
\end{Lemma}
\begin{proof}
    We only show this for MCNB spacetimes, the other cases are
    proven in exactly the same way by replacing all occurrences of
    ``causal" with ``timelike" resp.\ ``null".  Let $\gamma, \sigma$
    be two maximising, causal geodesics meeting tangentially at $p$,
    which is not an endpoint of either curve. W.l.o.g.\ we may assume that $\gamma(0)=\sigma(0)=p$. It suffices to show that there is $\eps>0$ such that $\gamma|_{[0,\eps)}=\sigma|_{[0,\eps)}$. Indeed since then $\gamma$ and $\sigma$ meet tangentially at $t=\eps$, the maximal such interval coincides with the maximal domain of definition (to the right and analogously to the left).
    
    Assume to the contrary that $\gamma|_{[0,\eps)} \neq \sigma|_{[0,\eps)}$ for any $\eps$. As the curves both satisfy the geodesic equation at $p$ and their tangents agree, so do their second derivatives. Thus $\gamma$ and $\gamma|_{(-\eps,0]}\cup \sigma|_{[0,\eps)}$ are both (unbroken) geodesics and by MCNB $\gamma$ cannot branch at $t=0$. So for every $\eps>0$ we have $\gamma|_{(0,\eps)}\cap \sigma|_{(0,\eps)}\neq \emptyset$.
    Hence there exist $t_k \downarrow 0$ such that $\gamma(t_k)=\sigma(t_k)$,
    and by our indirect assumption there exists $s_k \downarrow 0$ such that $\gamma(s_k) \neq \sigma(s_k)$. For any $k$ set 
    \begin{equation*}
        \eta_k:= \sup \{a\in[0,s_k]: \gamma|_{(s_k-a,s_k]} \cap \sigma|_{(s_k-a,s_k]}= \emptyset\}.
    \end{equation*}
    Further choose some $\tilde \eta_k>0$ such that for $I_k:=(s_k-\eta_k,s_k+\tilde \eta_k)$ we have $\gamma|_{I_k}\cap \sigma|_{I_k}= \emptyset$.
    Since $t_k \downarrow 0$ and $\gamma(t_k)=\sigma(t_k)$, we know that $s_k>\eta_k$. 
    However, $\gamma(s_k-\eta_k)=\sigma(s_k-\eta_k)$ by our choice of $\eta_k$. Set $s_k-\eta_k=r_k$ and for simplicity of notation assume $r_k<t_k$.
     If $\dot{\gamma}(r_k)\neq \dot{\sigma}(r_k)$, the curve $\alpha:= \gamma|_{[0,r_k]}\cup \sigma|_{[r_{k},t_k]}$ is a broken geodesic from $p$ to $\gamma(t_k)$ and hence not maximising by 
     \cite[Lem.\ 3.2]{SS21}, i.e.,
    $d(p,\gamma(t_k))> L(\alpha)$. Since $\gamma$ and $\sigma$ are maximisers we have
    $d(p,\gamma(r_{k}))= L(\gamma|_{[0,r_k]})=L(\sigma|_{[0,r_k]})$
    Consequently,
    \[
    d(p,\sigma(t_k))=L(\sigma|_{[0,t_k]})= L(\sigma|_{[0,r_k]})+L(\sigma|_{[r_k,t_k]})=L(\alpha)< d(p,\sigma(t_k)), 
    \]
    a contradiction, and hence $\dot{\gamma}(r_k)=\dot{\sigma}(r_k)$.
    
    Again employing the geodesic equation, also the second derivatives of $\gamma$ and $\sigma$ coincide at $r_k$, and so 
    the curves $\gamma|_{[0,s_k]}$ and $\gamma|_{[0,r_k]}\cup \sigma|_{[r_k,s_k]}$ display maximal causal branching, a contradiction.
\end{proof}

The following result shows that in maximally causally non-branching spacetimes, maximal causal geodesics can always be approximated by geodesics for 
the regularised spacetimes $\hat g_\eps$ and $\check g_\eps$.
Whether such a result is true in general is questionable and very likely regularisation dependent. 

\begin{Proposition}\label{prop:approxmcnb}
    Let $(M,g)$ be a $C^1$-spacetime that is MCNB.
    \begin{enumerate}
        \item Suppose that $M$ is globally hyperbolic and let $\varepsilon_k\searrow 0$ ($k\to\infty$). Set $g_k := \hat g_{\eps_k}$ or $g_k := \check g_{\eps_k}$ for all $k$. If
        $\gamma:[0,a]\to M$ is a maximising, timelike $g$-geodesic
        then for any small $\delta>0$ there exists a subsequence $g_{k_l}$ of $g_k$ and $g_{k_l}$-maximising, timelike geodesics $\gamma_l$ converging in $C^1$ to $\gamma|_{[0,a-\delta]}$.
        \item Let $g_k = \check g_{\eps_k}$ for all $k$. If $M$ is causal
        and $\gamma:[0,a]\to M$ is a $g$-maximizing null geodesic, then for
        any small $\delta>0$, there exists a subsequence $g_{k_l}$ of
        $g_k$, $t_l\downarrow 0$ and $g_{k_l}$-null geodesics
        $\gamma_l:[t_l, a-\delta]\to M$ contained in $\partial
        I_l^+(\gamma(0))$ (hence in particular
        $g_{k_l}$-maximising), which converge in $C^1_{\mathrm{loc}}$ to
        $\gamma|_{[0,a-\delta]}$.

        If, furthermore, $M$ is strongly causal and $S$ is an acausal set such that
        $\gamma \subseteq E^+(S)$, then even $\gamma_l:[t_l,a-\delta]\to E^+_l(S)$ and $\gamma_l(t_l) \in S$ with $\gamma_l(t_l)\to \gamma(0)$.
    \end{enumerate}
\end{Proposition}
\begin{proof}
Fix $\delta \in (0,a)$ and set $p:=\gamma(0)$, $q:=\gamma(a)$, and $q_\delta:= \gamma(a-\delta)$. 
\begin{enumerate}
    \item Since $\gamma$ is timelike, $q_\delta \in I^+(p)$ and without loss of generality we may assume that also $q_\delta \in I^+_k(p)$ for all $k$. 
    Since $g$ is globally hyperbolic we can also assume w.l.o.g.\ that $g_k$ are globally hyperbolic as well (see \cite[Prop.\ 2.3 (iv)]{penrosec11}, which remains valid for $g\in C^1$).
    Thus there exist $g_k$-maximising, timelike geodesics $\gamma_k$ from $p$ to $q_\delta$ \cite[Prop.\ 6.5]{S14}. By Corollary \ref{Corollary: Hartmangeodesicversion} an affinely reparametrised subsequence, without loss of generality $\gamma_k$ itself, converges in 
    $C^1_{\mathrm{loc}}$ to a $g$-causal geodesic $\sigma$. By Lemma \ref{Lemma: limitsofmaximizers} $\sigma$ is a $g$-maximising timelike geodesic from $p$ to $q_\delta$, in particular $d(p,q_\delta)=L(\gamma|_{[0,a-\delta]})= L(\sigma)$. We cannot yet apply Lemma \ref{Lemma: nonbranchingmeettangentiallyequal} since 
    it is not clear whether $\sigma$ is maximising beyond $q_\delta$. 
    However denote by $\eta := \sigma\cup \gamma|_{[a-\delta,a]}$ the concatenation of $\sigma$ and
    $\gamma$. Since $d(p,q)=L(\gamma)=L(\gamma|_{[0,a-\delta]})+L(\gamma|_{[a-\delta,a]})=
    L(\sigma)+L(\gamma|_{[a-\delta,a]})= L(\eta)$, it is a maximising, timelike geodesic from $p$ to $q$. We can now apply Lemma \ref{Lemma: nonbranchingmeettangentiallyequal} to the curves $\eta$ and $\gamma$, to obtain $\eta=\gamma$ and hence $\sigma=\gamma|_{[0,a-\delta]}$.

    \item  A similar statement was shown in
    \cite[Cor.\ 3.5]{SS21}, but we need to adapt some of the main points
    of the argument to our assumptions.
    
    Choose a sequence of points $q_\delta^k \in \pt I^+_{g_k}(p)$ converging to
    $q_\delta$. There exist future directed, $g_k$-null maximising geodesics
    $\gamma_k$ ending at $q_\delta^k$ and contained in $\pt I^+_{g_k}(p)$. 
    W.l.o.g\ we may assume their terminal velocities to be $h$-normalized,
    so we can apply Corollary \ref{Corollary: Hartmangeodesicversion}
    to obtain a subsequence of $\gamma_k$ (denoted in the same way)
    converging to a $g$-null geodesic $\sigma$. Since $\pt
    I^+_{g_k}(p)\sse \overline{I^+(p)}$ (recall that $g_k=\check
    g_{\varepsilon_k}$) we know that $\sigma \sse \pt I^+(p)$, since
    otherwise $q_\delta \in I^+(p)$, hence $\sigma$ is maximising (cf.\ Lemma \ref{Lemma: nulllines}). 
    Note that $\gamma$ and $\sigma$ meet tangentially at $q_\delta$,
    because otherwise $\eta:=\sigma\cup \gamma|_{[a-\delta,a]}$ would not be maximising by
    Lemma \ref{Lemma: maxunbrokengeodesic} and so $q \in I^+(\sigma)\subseteq I^+(\overline{I^+(p)})\subseteq I^+(p)$, a
    contradiction. 
    Lemma \ref{Lemma: nonbranchingmeettangentiallyequal} now shows that 
    $\gamma=\eta$ and so $\sigma=\gamma$, wherever both curves are defined. There are two possibilities for $\sigma$, 
    either it is past inextendible or it reaches $p$. However, if it
    is past inextendible, we must have $\sigma\supseteq
    \gamma|_{[0,a-\delta]}$. 
    Hence in both cases the $\gamma_k$ converge to
    $\gamma$ in $C^1_\text{loc}$. Further this means that there are $t_k\in
    [0,a-\delta]$ with $\gamma_k(t_k)\to p$, and we have $t_k \to t_0= 0$ since
    otherwise $\gamma(t_0)=p=\gamma(0)$, contradicting the causality of $M$.

    Now assume $M$ to be strongly causal and $\gamma\subseteq E^+(S)$ for $S$ acausal, then
    there exists a $g$-causally convex, $g$-globally hyperbolic
    neighbourhood $U$ of $p$ by \cite[Lem.\ 3.21]{MinSan08} (this result
    holds also for $C^1$-spacetimes). In such neighbourhoods it holds that
    $I_U^+(S)=I^+(S)\cap U$, $J_U^+(S)=J^+(S)\cap U$ and hence also
    $E_U^+(S)=\pt I_U^+(S)=\pt I^+(S)\cap U$. Clearly $U$ is also
    $g_k$-causally convex and $g_k$-globally hyperbolic and so these
    equalities also hold for $g_k$. Note that in the same fashion as for the case of a single point $p$
    one can show that $\gamma_k \to \gamma|_{[0,a-\delta]}$. If $\gamma_k$
    reaches $S$, we immediately obtain $\gamma_k \subseteq 
    E^+_k(S)$. In the other case $\gamma_k\subseteq \pt I^+_k(S)\backslash
    J^+_k(S)$ is past inextendible. However, as they converge to 
    $\gamma|_{[0,a-\delta]}$, there must be $p_k\in \gamma_k$, with 
    $p_k\to p$. Hence, w.l.o.g.\ $p_k \in \pt I_k^+(S)\cap U= E_{U,k}^+(S)$ and so $p_k\in J^+_k(S)$, a contradiction.

    We have shown that $\gamma_k$ reach $S$, and it only remains to show that they
    do so close to $0$. Again let $\gamma_k(t_k)\in S$ with $\gamma_k(t_k)
    \to p$. If $t_k\to t_0$, then $\gamma(t_0)\in S$ and by acausality of $S$ we must have $t_0=0$.
    \end{enumerate}
\end{proof}

\begin{remark}\label{remark:approxmcnb} 
\begin{itemize}
\item[(i)] As the proof of Proposition \ref{prop:approxmcnb} shows, for (i) it suffices to assume $(M,g)$ to be MTNB, while for (ii) MNNB
would suffice.
\item[(ii)] We will make use of Proposition \ref{prop:approxmcnb}(ii)
in two different scenarios:  Thm.\ \ref{Theorem: nonulllines} and Proposition
\ref{Proposition: lightraysfromsubmfd} below.
In the first case ($M$ 
causal) all we need to ensure is that $\gamma_l$ are maximizing for long
enough, i.e.\ $a-\delta-t_l$ is close to $a$. In the second ($M$ strongly causal) it is
important to also obtain that $\gamma_l$ reaches $S$, in order to be able
to use the initial conditions assumed on $S$.
\end{itemize}
\end{remark}

\section{No lines}
\label{sec:max_geods}

In this section we will prove that, under suitable causality and energy conditions, complete causal geodesics stop being maximising also in (MCNB) $C^1$-spacetimes. To this end we first prove the existence of conjugate points for causal geodesics in smooth spacetimes under weakened versions of the energy conditions. More precisely we will invoke the energy conditions derived for regularisations in Lemma \ref{Lemma: timelikeenergyconditionapproxmetrics1} and in Lemma \ref{Lemma: nullenergyconditionapproxmetrics} from the distributional energy conditions, as well as in Lemma \ref{Lemma: tidalforcematrixestimateapprox} from the distributional genericity condition.
The following is a strengthened version of \cite[Prop.\ 4.2]{GGKS}:

\begin{Lemma} 
\label{Lemma: conjpointssmoothmetric}
    Let $(M,g)$ be a smooth spacetime. Given some $c>0$ and $0<r< \frac{\pi}{4 \sqrt{c}}$,
    there exist
    $\delta(c,r)>0 $ and $T(c,r)>0$ such that any causal geodesic $\gamma$ defined on $[-T,T]$ for which
    the following hold
    \begin{itemize}
        \item[(i)] $\text{Ric}(\dot{\gamma},\dot{\gamma})\geq - \delta $ on $[-T,T]$ and 
        \item[(ii)] there exists a smooth parallel orthonormal frame for $[\dot{\gamma}]^\perp$ and some
        $C>0$ such that w.r.t.\ this frame the tidal force operator satisfies 
        $[R](t) > \mathrm{diag}(c,-C,\ldots,-C)$ on $[-r,r]$,
    \end{itemize}
    possesses a pair of conjugate points on $[-T,T]$.
\end{Lemma}

\begin{proof} This actually follows by a more careful choice of constants in the proof of \cite[Prop.\ 4.2]{GGKS}. To see this 
we briefly recall the main steps of that proof. The argument proceeds indirectly, assuming
that for any $\delta>0$ and $T>0$ there is some $\gamma$ satisfying (i) and (ii) without conjugate points in $[-T,T]$. 
Denote by $[A]$ the unique Jacobi tensor class along $\gamma$ with $[A](-T)=0$ and $[A](0)=\mathrm{id}$. 
Taking $[E_1](t),\dots,$ $[E_{d}](t)$ as in (ii), linear endomorphisms
of $[\dot \gamma]^\perp$ are written as matrices in this basis, and we set
$[\tilde R](t):=\diag(c,-C,\dots,-C)$ with $C$ as in (ii). Then by (ii), $[\tilde R](t)<[R](t)$ on $[-r,r]$.

The self-adjoint operator $[B]:=[\dot A]\cdot [A]^{-1}$satisfies the matrix Riccati equation
\begin{equation}\label{riccati} 
[\dot B] + [B]^2 + [R] = 0, 
\end{equation}
and we denote by $[\tilde B]$ the solution to~\eqref{riccati}, with $[\tilde R]$ instead of
$[R]$ and initial value prescribed at some $t_1\in [-r,r]$. 
We may even choose $t_1$ in $[-r,0]$ and $[\tilde B](t_1) := 
\tilde \beta(t_1)\cdot \mathrm{id}$, where $\tilde \beta(t_1)$ is greater or
equal than the largest eigenvalue of $[B](t_1)$. More precisely, examining the Raychaudhuri equation
\begin{equation}
\dot \theta + \frac{1}{d}\theta^2 + \mathrm{tr}(\sigma^2) + \mathrm{tr}([R]) =
0,
\end{equation}
for the expansion $\theta=\mathrm{tr}([B])$ (with $\sigma=[B]-\frac{1}{d}\theta\cdot
\mathrm{id}$), it follows that one can pick $\tilde \beta(t_1):=f(\nu,\delta,r)$ and $[\tilde
B](t_1):=f(\nu,\delta,r)\cdot\mathrm{id}$ to indeed achieve that $[B](t)\le
[\tilde B](t)$ on $[t_1,r]$. Here, $f(\nu,\delta,r) = \sqrt{\frac{2\nu}{r}+\delta} + \frac{\nu}{d}$
with $\nu=4d/T$.

Moreover, due to the fact that both 
$[\tilde R]$ and $[\tilde B](t_1)$ are diagonal, the Riccati
equation for $[\tilde B]$ decouples.  Indeed it can be explicitly solved by
\[
[\tilde B](t) = \frac{1}{d}\diag(H_{c,f}(t),H_{-C,f}(t),\dots,H_{-C,f}(t)),
\]
where
\[
H_{c,f}(t) = d \sqrt{c}\cot(\sqrt{c}(t-t_1) + \mathrm{arccot}(f/\sqrt{c})),
\]
and
\[
H_{-C,f}(t)=d \sqrt{C}\tanh\big(\sqrt{C}(t-t_1)+\mathrm{artanh}(f/\sqrt{C})\big).
\]
From this point on, the proof of \cite[Prop.\ 4.2]{GGKS} proceeds by exclusively 
analysing the function $H_{c,f}$, and the only requirement on $r$ turns out to be $4r\sqrt{c}<\pi$.
Once this is granted, $\delta$ and $T$ can be chosen depending only on $c$ to arrive at the
desired contradiction. 
\end{proof}

Recall that a line is an inextendible causal geodesic maximising the Lorentzian distance between any of its points.

\begin{Theorem}(No timelike lines)\label{Theorem: notimelikelines}

\noindent Let $(M,g)$ be a globally hyperbolic, MTNB $C^1$-spacetime satisfying the distributional timelike energy condition and the distributional genericity condition along any inextendible timelike geodesic. Then there is no complete timelike line in $M$.
\end{Theorem}

\begin{proof}
Suppose $\gamma:\R \to M$ is a complete timelike line. We may assume that distributional genericity
holds at $t_0=0$ along $\gamma$. We approximate $g$ by $\check{g}_{\varepsilon} \prec g$, which are hence also globally hyperbolic. 
By Lemma \ref{Lemma: tidalforcematrixestimateapprox} 
there exist $c > 0$, $0 < r < \frac{\pi}{4\sqrt{c}}$ with the following property: For any $C^1$-approximation
$\gamma_{\varepsilon}$ of $\gamma$ by $\check{g}_{\varepsilon}$-timelike geodesics, there exists
$\varepsilon_0 > 0$ such that for all $\varepsilon < \varepsilon_0$, there is $C=C(\varepsilon)>0$
so that (ii) in Lemma \ref{Lemma: conjpointssmoothmetric} is satisfied for $R_{\varepsilon}$. Now for the above choice of $c$ and $r$, pick
$\delta > 0$ and $T > 0$ as in Lemma \ref{Lemma: conjpointssmoothmetric}, and let $\tilde{T} > T$. By
Proposition \ref{prop:approxmcnb}(i), for the curve $\gamma|_{[-T,\tilde T]}$ and $\delta< \tilde T-T$, we obtain a subsequence
$\check g_{\eps_k}$ and $\check g_{\eps_k}$-maximal timelike geodesics $\gamma_k$ from $\gamma(-T)$
to $\gamma(\tilde T)$ converging to $\gamma|_{[-T,\tilde T]}$ in $C^1$.

Let $K$ be a compact neighbourhood of  
$\gamma([-T,T])$. Since $\gamma_{\varepsilon_k} \to \gamma$ in $C^1([-T,T])$, there are $k_0 \in \mathbb{N}$, $\tilde{C} > 0$ and $\kappa < 0$ such that for all $k \geq k_0$, we have $\gamma_{\varepsilon_k}([-T,T]) \subseteq K$, $\|\dot{\gamma}_{\varepsilon_k}\|_h \leq \tilde{C}$ and $\check g_{\varepsilon_k}(\dot{\gamma}_{\varepsilon_k},\dot{\gamma}_{\varepsilon_k}) < \kappa$ on $[-T,T]$. Hence by \cite[Lem.\ 4.6]{G20} (and the remark preceding it) 
we have $\Ric_{\varepsilon_k}(\dot{\gamma}_{\varepsilon_k},\dot{\gamma}_{\varepsilon_k})\geq - \delta$ on $[-T,T]$ for large $k$. Therefore, also (i) in Lemma \ref{Lemma: conjpointssmoothmetric} is
satisfied for $\gamma_{\eps_k}$, yielding the existence of conjugate points for any $\gamma_{\varepsilon_k}$ ($k$ large) in $[-T,T]$, a contradiction since they are maximising even in the strictly larger interval $[-T,\tilde{T}]$.
\end{proof}

\begin{Theorem}(No null lines)
\label{Theorem: nonulllines}

\noindent Let $(M,g)$ be a causal, MNNB 
$C^1$-spacetime that satisfies the distributional null energy condition and the distributional genericity condition along any inextendible null geodesic. Then there is no complete null line in $M$.
\end{Theorem}
\begin{proof}
Suppose $\gamma:\R \to M$ is a complete null line. We will again make use of the approximation
$\check{g}_{\varepsilon}$. Suppose without loss of generality that the distributional genericity condition holds along
$\gamma$ at $t_0=0$. By Lemma \ref{Lemma: tidalforcematrixestimateapprox} there exist $c > 0$ and $0 < r <
\frac{\pi}{4\sqrt{c}}$ satisfying: For any $C^1$-approximation $\gamma_{\varepsilon}$ of $\gamma$ by
$\check{g}_{\varepsilon}$-null geodesics, there exists $\varepsilon_0 > 0$ such that for all
$\varepsilon < \varepsilon_0$, there is $C=C(\varepsilon)>0$ so that (ii) in Lemma \ref{Lemma:
conjpointssmoothmetric} is satisfied for $R_{\varepsilon}$. Choose $\delta > 0$ and $T > 0$ as in Lemma
\ref{Lemma: conjpointssmoothmetric} for this pair $(c,r)$ and let $\tilde{T} > T$. Using Proposition
\ref{prop:approxmcnb}(ii) for the curve $\gamma|_{[-T,\tilde T]}$ and $S =\gamma(-T)$ (cf.\ Remark \ref{remark:approxmcnb}(ii)), there exists a subsequence $\check
g_{\eps_k}$ and $\check g_{\eps_k}$-maximal null geodesics $\gamma_k: [-T,\tilde T] \to M$ 
converging to $\gamma|_{[-T,\tilde T]}$ in $C^1$.

Since $\gamma_k \to \gamma$ in $C^1([-T,T])$, once we choose a compact neighbourhood $K$ of $\gamma([-T,T])$, there are $k_0 \in \mathbb{N}$ and $\tilde{C}_2 > \tilde{C}_1 > 0$ such that for all $k \geq k_0$, we have $\gamma_{\varepsilon_k}([-T,T]) \subseteq K$, $\tilde{C}_1 < \|\dot{\gamma}_k\|_h < \tilde{C}_2$ in $[-T,T]$. Hence, Lemma \ref{Lemma: nullenergyconditionapproxmetrics} implies that $\Ric_{\varepsilon_k}(\dot{\gamma}_k,\dot{\gamma}_k) > - \delta$ in $[-T,T]$. But then Lemma \ref{Lemma: conjpointssmoothmetric} may be invoked to give the existence of conjugate points on $\gamma_k|_{[-T,T]}$ for large $k$, a contradiction since they were supposed to be maximising even on $[-T,\tilde{T}]$.
\end{proof}

\section{Initial conditions} 
\label{sec:trapped}

The classical Hawking-Penrose theorem for spacetimes of dimension $4$ assumes three alternative initial conditions: the existence of a compact, spacelike hypersurface, a trapped ($2$-)surface, or a ``trapped point'', i.e., a point from where all future (or past) null geodesics converge. The second condition was later  generalised in \cite{GS} to trapped submanifolds of arbitrary codimension $m$ with $1<m<n=\dim(M)$. 

While the first condition does not need any special attention here we will start by generalising the trapped submanifold-case to $C^1$-spacetimes. As for $C^{1,1}$-spacetimes (see \cite[Sec.\ 6.2]{GGKS}), trapped submanifolds are defined in the support sense. However, to show that normal null geodesics emanating from them stop being maximising is more delicate now due to the lack of an exponential map, cf.\ (the proof of) Proposition \ref{Proposition: lightraysfromsubmfd}. At the end of this section we will deal with the case of a ``trapped point''.

\begin{definition}(Support submanifolds)
\label{definition: supportsubmfd}

\noindent Let $(M,g)$ be a $C^1$-spacetime and let $S,\tilde{S} \subseteq M$ be submanifolds. We say that $\tilde{S}$ is a \textit{future support submanifold for $S$ at $q \in S$} if $\dim S = \dim \tilde{S}$, $q \in \tilde{S}$, and there is a neighbourhood $U$ of $q$ in $M$ such that $\tilde{S} \cap U \subseteq J^+_U(S)$.
\end{definition}

\begin{definition}(Trapped $C^0$-submanifolds)
\label{definition: trappedc0submfd}

\noindent Let $(M,g)$ be a $C^1$-spacetime of dimension $n$. We say that a $C^0$-submanifold $S
\subseteq M$ of codimension $1 \le m<n$ 
is a \textit{future trapped submanifold} if it is compact
without boundary and for any $p \in S$ there exists a neighbourhood $U$ of $p$ such that 
$U \cap S$ is achronal in $U$ and moreover $S$ has past pointing timelike mean curvature 
in the sense of support submanifolds, i.e.\ for any $q \in S$ there exists a future
$C^2$-support submanifold $\tilde{S}$ for $S$ at $q$ whose mean curvature vector at $q$ is
past-pointing timelike.
\end{definition}

Now, to show that lightrays from trapped submanifolds stop maximising, we reduce the problem to a question about smooth approximating metrics, which are understood much better. The essential results that deal with the corresponding situations for smooth metrics are \cite[Lem.\ 6.4]{GGKS} and \cite[Lem.\ 5.6]{G20}. To give a precise formulation we first introduce some notation.

Suppose $S \subseteq M$ is a spacelike $C^2$-submanifold of codimension $1<m<n$, 
$p \in S$ and $\nu \in T_p M$ is a future null normal to $S$. Let $\gamma$ be
a geodesic with $\gamma(0) = p$, $\dot\gamma(0)=\nu$. Let $e_1,\dots,e_{n-m}$ be an
ON-basis on $S$ around $p$. Further let $E_1,\dots,E_{n-m}$ be the parallel translates of
$e_1(p),\dots,e_{n-m}(p)$ along $\gamma$, which are $C^1$ since they satisfy the parallel
transport equation. We will use this notation for the following results. It will be
clear from the context which surface the $E_i$ refer to.

\begin{Lemma}\label{Lemma: smoothcodimnot2} Let $(M,g)$ be a smooth spacetime and let $S$
be a codimension-$m$ $(1<m<n)$ spacelike $C^2$ submanifold of $M$. Let $\gamma$ be a
geodesic starting in $S$ such that $\nu:=\dot{\gamma}(0) \in TM|_S$ is a future null normal
to $S$. Suppose that $c:=\mathbf{k}_S(\nu) > 0$ and let $b>1/c$. Then there is $\delta=\delta(b,c) > 0$ such that, if for $E_i$ as described above
\begin{align*}
    \sum_{i=1}^{n-m} g(R(E_i,\dot{\gamma})\dot{\gamma},E_i) \geq - \delta
\end{align*}
along $\gamma$, then $\gamma|_{[0,b]}$ is not maximising from $S$, provided that $\gamma$ exists up to $t=b$.
\end{Lemma}

\begin{Lemma}\label{Lemma: smoothcodim2} Let $(M,g)$ be a smooth spacetime and let $S$ be a codimension-$2$ spacelike $C^2$ submanifold of $M$. Let $\gamma$ be a geodesic starting in $S$ such that $\nu:=\dot\gamma(0)$ is a future null normal to $S$. Suppose that $c:=\mathbf{k}_S(\nu) > 0$ and let $b>1/c$. Then there is $\delta = \delta(b,c) > 0$ such that, if
\begin{align*}
    \Ric(\dot{\gamma},\dot{\gamma}) \geq -\delta
\end{align*}
along $\gamma$, then $\gamma|_{[0,b]}$ is not maximising from $S$, provided that $\gamma$ exists up to $t=b$.
    
\end{Lemma}

The following result is the $C^1$-analogue of \cite[Prop.\ 6.5]{GGKS}. As was the case for the distributional genericity condition, cf.\ Definition \ref{definition: strongdistributionalgenericity}, also here we have to assume that the distributional curvature condition at hand is stable under $C^0$-perturbations of the vector fields involved. This ensures that we can derive the necessary curvature conditions for the smooth approximating metrics in order to be able to use Lemmas \ref{Lemma: smoothcodimnot2} and \ref{Lemma: smoothcodim2}.

\begin{Proposition}(Light rays from a submanifold)
\label{Proposition: lightraysfromsubmfd}

\noindent Let $(M,g)$ be a strongly causal, MNNB $C^1$-spacetime and let
$\tilde{S} \subseteq M$ be a $C^2$-spacelike submanifold of codimension $1<m<n$. Suppose
$\mathbf{k}_{\tilde{S}}(\nu) = g(\nu,H_p) > c > 0$ and let $b>1/c$. 
Suppose there is a null geodesic $\gamma$ with $\dot{\gamma}(0)=\nu$, a neighbourhood
$U$ of $\gamma|_{[0,b]}$ and $C^1$-extensions $\overline{E}_i$ of $E_i$ and $\overline{N}$
of $\dot{\gamma}$ to $U$ such that for each $\delta >0$ there exists $\eta > 0$ such
that for all collections of $C^1$-vector fields
$\{\tilde{E}_1,\dots,\tilde{E}_{n-m},\tilde{N}\}$ on $U$ with $\|\tilde{E}_i -
\overline{E}_i\|_h < \eta$ for all $i$ and $\|\tilde{N}-\overline{N}\|_h < \eta$, we
have
\begin{align}\label{eq:initialcurvature}
    \sum_{i=1}^{n-m} g(R(\tilde{E}_i,\tilde{N})\tilde{N},\tilde{E}_i) \geq -\delta \quad \text{in } \mathcal{D}^{'(1)}(U).
\end{align}
Then $\gamma|_{[0,b]}$ is not maximising from $\tilde{S}$.
\end{Proposition}

Note that condition (\ref{eq:initialcurvature}) localises in a similar manner as the genericity condition. For the latter see Remark \ref{rem:localizegen}.
\begin{proof}
Let $g_{\varepsilon} = \check{g}_{\varepsilon}$,
$R_{\varepsilon}=R[\check{g}_{\varepsilon}]$. Clearly, $\mathbf{k}_{\Tilde{S}}$ is continuous and
$\mathbf{k}_{\Tilde{S}, \varepsilon} \to \mathbf{k}_{\Tilde{S}}$ uniformly on compact sets. Hence, there is 
a neighbourhood $V$ of $\nu$ in $TM|_{\Tilde{S}}$ and $\varepsilon_0$ such that
$\forall \varepsilon \leq \varepsilon_0$ and $\forall v \in V$:
$\mathbf{k}_{\Tilde{S},\varepsilon}(v) > c$.
Note that $g$-spacelikeness implies $g_{\varepsilon}$-spacelikeness. Hence $U \cap
\Tilde{S}$ is $g_{\varepsilon}$-spacelike and we may assume that $W := \pi(V)$ is contained
in $U \cap \Tilde{S}$. 

Suppose now to the contrary that $\gamma|_{[0,b]}$ maximises the distance from $\overline{U}\cap\Tilde{S}$. Let
$1/c<b'<b''<b$ and let $q=\gamma(b'')$.
Find $q_k \in \partial J_k^+(\overline{U} \cap \Tilde{S})$
(where $\partial J_k:=\partial J_{g_{\varepsilon_k}}$ and $\varepsilon_k\to 0$) with $q_k \to q$. By Corollary \ref{Corollary: curvesinlightcone} there are $g_{\varepsilon_k}$-null geodesics $\gamma_k: I_k \to M$ (in future directed parametrisation) with $\gamma_k(b'') = q_k$ either intersecting $\overline{U} \cap \Tilde{S}$ or past inextendible. We may assume that $\|\dot{\gamma}_k(b'')\|_h$ are all equal to $\|\dot{\gamma}(b'')\|_h$, hence we may assume that the $\dot{\gamma}_k(b'')$ converge to a $g$-null vector $v$ and the geodesics converge to a $g$-null geodesic $\gamma_v$ in $C^2_{\mathrm{loc}}$ by Corollary \ref{Corollary: Hartmangeodesicversion}. Strong causality allows us to invoke Lemma \ref{Lemma: equivalentdefofstrongcausality}, so we
can find a neighbourhood $W$ of $\gamma(0)$ with $W \subseteq U$ such that $W \cap \tilde S$ is acausal in $M$. 
By (the proof of) Proposition \ref{prop:approxmcnb}(ii) (see also Remark \ref{remark:approxmcnb}(i)) we obtain that $\gamma = \gamma_v$, and that
(up to picking a subsequence) $\gamma_k$ reaches $\tilde S$ for $k$ large at $t_k \downarrow 0$. 

Since any of the $g_{\eps_k}$-geodesics $\gamma_k$ reaches $\tilde S$
for all $\eps_k<\eps_0$ (for some suitable $\eps_0>0$) at points $\gamma_k(t_k) \in
\tilde S$, by $C^1_{\mathrm{loc}}$ convergence of $\gamma_k$ to $\gamma$
we also have $\dot{\gamma}_k(t_k)\in V$.
Note that, due to $\gamma_k(t_k)\to \gamma(0)$, we can pick $g_{\eps_k}$-orthonormal bases for
$T_{\gamma_k(t_k)}\tilde S$ converging to a $g$-orthonormal basis for $T_{\gamma(0)}\tilde S$.
Denote by $E_i^{\eps_k}$ and $E_i$ the $g_{\eps_k}$- and $g$-parallel
transports of these bases along $\gamma_k$ and $\gamma$, respectively. 
By Lemma \ref{lem:A2} there exist $C^1$-extensions
$\tilde E_i^{\eps_k}$ and
$\tilde{E_i}$ to the (possibly shrunk) neighbourhood $U$. By a
similar argument one can extend the velocity vector fields 
along $\gamma_k$ and $\gamma$ to all of $U$, they will be denoted by
$\tilde{ N}^\eps$ and $\tilde{N}$, respectively.
 
Assume also without loss of generality that $U$ is relatively compact and that $\gamma_k([t_k,b])
\subseteq U$ for large $k$.
Pick $\delta=\delta(b',c)$ 
according to Lemma \ref{Lemma: smoothcodimnot2}.
Since they restrict to the same vector fields along $\gamma$, 
by shrinking $U$ and by continuity of $\overline{E_i}$ and $\tilde E_i$ resp.\
$\overline{N}$ and $\tilde N$,
we can assume these vector fields to be
arbitrarily close in $\|\,.\,\|_h$ on $U$. Moreover, by the convergence of
$\tilde E_i^{\varepsilon_k}$ and $\tilde N^{\varepsilon_k}$ to
$\tilde E_i$ and $\tilde{N}$, respectively, also these vector fields can be made 
arbitrarily $\|\,.\,\|_h$-close to $\overline{E_i}$ and $\overline{N}$, respectively,
on $U$ for large $k$. The assumption of the Proposition  therefore implies that
\begin{align*}
    \sum_{i=1}^{n-m} g(R(\tilde E_i^{\varepsilon_k},\tilde{N}^{\varepsilon_k})\tilde{N}^{\varepsilon_k},\tilde E_i^{\varepsilon_k}) \ge -\delta/2
\end{align*}
in  $\D'(U)$ for $k$ large. Since $\star_M \rho_{\varepsilon_k}$ respects inequalities, we conclude that also
\begin{align*}
    \sum_{i=1}^{n-m} g(R(\tilde E_i^{\varepsilon_k},\tilde{N}^{\varepsilon_k})\tilde{N}^{\varepsilon_k},\tilde E_i^{\varepsilon_k}) \star_M \rho_{\varepsilon_k} \ge -\delta/2
\end{align*}
on $U$ for large $k$.

By the same reasoning as in Lemma \ref{Lemma: genericityfriedrichs} we obtain (once again writing $R_{\varepsilon_k}$ for $R[g_{\varepsilon_k}]$)
\[
g(R(\tilde E_i^{\varepsilon_k},\tilde{N}^{\varepsilon_k})\tilde{N}^{\varepsilon_k},\tilde E_i^{\varepsilon_k}) \star_M \rho_{\varepsilon_k} -
g_{\varepsilon_k}(R_{\varepsilon_k}(\tilde E_i^{\varepsilon_k},\tilde{N}^{\varepsilon_k})
\tilde{N}^{\varepsilon_k},\tilde E_i^{\varepsilon_k}) \to 0
\] 
uniformly on $U$ as $\varepsilon_k \to 0$, $i=1,\dots,n-m$.
But then, for large $k$,
\begin{align*}
 \sum_{i=1}^{n-m}   g_{\varepsilon_k}(R_{\varepsilon_k}(\tilde E_i^{\varepsilon_k},\tilde{N}^{\varepsilon_k})
\tilde{N}^{\varepsilon_k},\tilde E_i^{\varepsilon_k}) \ge -\delta.
\end{align*}
In particular, the above holds along $\gamma_k$, so
\begin{align*}
  \sum_{i=1}^{n-m}  g_{\varepsilon_k}(R_{\varepsilon_k}(E_i^{\varepsilon_k}(t),\dot{\gamma}_k(t))\dot{\gamma}_k(t),E_i^{\varepsilon_k}(t)) \ge -\delta
\end{align*}
for all $t \in [t_k,b']$, where we may assume that $k$ is so large that all $\gamma_k$ are defined on $[t_k,b']$.

Due to our choice of $\delta$, Lemma \ref{Lemma: smoothcodimnot2} implies that, for $k$ sufficiently large, 
each $\gamma_k$  
stops maximising at parameter $b'+t_k$ at the latest (if $\gamma_k$ is not
$g_{\varepsilon}$-normal to $\tilde{S}$, then $\gamma_k$ stops maximising the distance immediately, see Remark \ref{remark: notnullnormal} below).
By construction the $\gamma_k$ maximize from $t_k$ to $b''$, hence if $k$ is so large that $b'+t_k<b''$ we obtain a contradiction.

\end{proof}

\begin{remark}
By arguments analogous to those given for the genericity condition, cf.\ Remark
\ref{remark: distribgenericityequivalentinC11}, namely essentially by boundedness of $R$
on compact sets, one can see that if $g \in C^{1,1}$, then the curvature condition \eqref{eq:initialcurvature} in Proposition
\ref{Proposition: lightraysfromsubmfd} is equivalent to the one given in \cite[Prop.\
6.5]{GGKS}.
\end{remark}

\begin{Proposition}(The case of trapped surfaces)
\label{Proposition: caseoftrappedsurfaces}

\noindent Let $(M,g)$ be a strongly causal, MNNB $C^1$-spacetime satisfying the distributional null energy condition. Let $\tilde{S} \subseteq M$ be a codimension-$2$ $C^2$-spacelike submanifold such that $\mathbf{k}_{\tilde{S}}(\nu) > c > 0$, where $\nu$ is a null normal to $\tilde{S}$, and let $b>1/c$. If $\gamma$ is a null geodesic with $\dot{\gamma}(0) = \nu$, then $\gamma|_{[0,b]}$ is not maximising from $\tilde{S}$.
\end{Proposition}
\begin{proof}
The proof is completely analogous to that of Proposition \ref{Proposition: lightraysfromsubmfd}, with the sole difference that one
needs to use Lemma \ref{Lemma: smoothcodim2} instead of Lemma \ref{Lemma: smoothcodimnot2}. Let us give a brief outline:
Approximate $g$ again by $g_{\eps_k} \equiv\check{g}_{\varepsilon_k}$ and the compact segment $\gamma|_{[0,b]}$ as
before by $g_{\eps_k}$-null geodesics $\gamma_k$. Then by Lemma \ref{Lemma: nullenergyconditionapproxmetrics} (see also the proof of
\cite[Thm.\ 5.7]{G20}), for all $\delta > 0$ and large $k$ we have
$\Ric_{\eps_k}(\dot \gamma_k,\dot\gamma_k) > - \delta$.
Using this, combined with Lemma \ref{Lemma: smoothcodim2} and proceeding as in Proposition \ref{Proposition: lightraysfromsubmfd}, one arrives at a contradiction as before.
\end{proof}

\begin{remark}
    The corresponding $C^{1,1}$-version of Proposition
    \ref{Proposition: caseoftrappedsurfaces} was dealt with in \cite[Remark 6.6(i)]{GGKS},
    however, a similar line of reasoning is not available in our case of $C^1$-metrics:
    While it is still possible to extend an ON-frame along a curve to a neighbourhood,
    this does not suffice to derive condition \eqref{eq:initialcurvature} of Proposition 
    \ref{Proposition: lightraysfromsubmfd}, as the needed rigidity of the condition cannot be
    guaranteed. More precisely: For any ON-frame $\{\bar E_i, \bar N \}$ on $U$ the equality (as before
    $\bar N$ denotes the extension of $\gamma'$ in the frame)
    \begin{equation*}
        \sum_{i=1}^{n-m} g(R(\bar{E}_i,\bar{N})\tilde{N},\bar{E}_i) = \Ric[g](\bar N, \bar N) 
        \geq - \delta
    \end{equation*}
    still holds and the final inequality is due to the distributional null energy condition.
    However, to carry out a proof similar to the one of Proposition \ref{Proposition: lightraysfromsubmfd}, this
    estimate would need to hold for all collections of vector fields close to the frame in the sense
    given in the Proposition. This cannot be guaranteed for $C^1$-metrics anymore as $R$ is, in
    general, unbounded. So instead, the slightly different focusing result used in Proposition
    \ref{Proposition: caseoftrappedsurfaces} is required.
    \end{remark}

\begin{remark}
\label{remark: notnullnormal}
If, in the situation of Proposition \ref{Proposition: lightraysfromsubmfd}, $\nu$ is null but not a null normal, then geodesics in that direction immediately enter $I^+$: This is shown in the same way as in the smooth case, e.g.\ by using ideas from \cite[Lem.\ 10.45, Lem.\ 10.50]{ON83}: Let $c$ be a geodesic in direction $\nu$, hence in particular $C^2$. 
By \cite[Lem.\ 3.1]{SS21} for a (piecewise) $C^1$-vector field $X$ along $c$ with
$g(X',\dot{c})<0$, any variation $c_s$ with variation field $X$ is timelike and longer than $c$
(for small $s$).

Now for any vector $v$ at $p$ not a null normal to a spacelike surface $S$, w.l.o.g.\ there exists $y\in T_pS$ with $g(y,v)> 0$. Define $V(t) =(1- \frac{t}{b}) Y(t)$, where $Y$ is the parallel translate of $y$ along $c$.
A straightforward calculation shows $g(V',\dot{c})<0$. Now choose a variation $c_s$ of $c$ with variational vector field $V$ such that $c_s(0)\in S$. We obtain that $c_s$ is timelike from $S$ to $c(b)$ and as $b$ was arbitrary the statement follows.
\end{remark}

The main result on trapped submanifolds in $C^1$-spacetimes now is the following:

\begin{Proposition}
\label{Proposition: submfdnullcompletenessortrappedset}
Let $(M,g)$ be a strongly causal, MNNB $C^1$-spacetime and let $S
\subseteq M$ be a trapped $C^0$-submanifold of codimension $m$, $1<m<\dim(M)=n$. If $m = 2$,
suppose that $(M,g)$ satisfies the distributional null energy condition, and if $m \neq 2$,
suppose that any support submanifold $\tilde{S}$ of $S$ satisfies the condition in
Proposition \ref{Proposition: lightraysfromsubmfd} in any null normal direction $\nu$ and for any null geodesic $\gamma$ with $\dot{\gamma}(0)=\nu$. Then $E^+(S)$ is
compact or $(M,g)$ is null geodesically incomplete.
\end{Proposition}
\begin{proof}
Suppose $(M,g)$ is null geodesically complete.\\
\textit{$E^+(S)$ is relatively compact}: Assume, to the contrary, that there are $q_k \in E^+(S)$
with $q_k \to \infty$. Lemma \ref{Lemma: nullgeodinlightcone} guarantees the existence of future directed maximising null geodesics $\gamma_k:[0,t_k] \to M$ connecting $p_k \in S$ to $q_k$, which we may assume to be  parametrised by $h$-unit speed. By compactness we may further assume that $\gamma_k(0)=p_k \to p \in S$, and we write $\gamma_k(t_k)= q_k.$
Due to $q_k \to \infty$, only finitely many $\gamma_k$ are contained in a fixed neighbourhood $U$ of $p$. So we may apply Theorem \ref{Theorem: limitcurvetheorem}, implying that (a subsequence, not relabeled, of) $\gamma_k$ converges to a $g$-null curve $\gamma:[0,\infty)\to M$ (in $h$-unit speed parametrisation) from $p$. Since the
$\gamma_k$ are $S$-maximising, so is $\gamma$. In particular, $\gamma$ is a $g$-null pre-geodesic. Reparametrising $\gamma$ as a geodesic, it is still
defined on $[0,\infty)$ due to null geodesic completeness and so $\gamma:[0,\infty) \to M$
is a future complete $g$-null $S$-ray. 
Let $\Tilde{S}$ be a future support submanifold of $S$
at $p$, then $\gamma$ maximises the distance from $\Tilde{S}$ everywhere, because
otherwise there would be some $T$ such that $\gamma(T) \in I^+(\Tilde{S}) \subseteq
I^+(S)$, contradicting the fact that $\gamma$ is an $S$-ray. As this cannot happen 
due to Proposition \ref{Proposition: lightraysfromsubmfd} resp.\ 
Proposition \ref{Proposition: caseoftrappedsurfaces}, $E^+(S)$ is relatively compact.

\textit{$E^+(S)$ is closed}: Let $q_k \in E^+(S)$, $q_k \to q$. Let $\gamma_k:[0,t_k]
\to M$ be null pre-geodesics (in $h$-unit speed parametrisation) from some $p_k \in S$ to
$q_k$. By compactness of $S$, we may assume that $p_k \to p \in S$. If $q = p$, then we are done. Otherwise, there are two
possibilities: If the $t_k$ are bounded, then by going over to subsequences, $t_k \to t \in (0,\infty)$. Since $q \neq p$, we are in a position to invoke Theorem \ref{Theorem: limitcurvetheoremtwoconvergingendpoints} to get a maximising
causal limit geodesic $\gamma:[0,t] \to M$ connecting $p$ to $q$. $\gamma$ cannot be timelike, as this would imply $q \in I^+(S)$ and hence $q_k \in I^+(S)$ for large $k$. 
Hence $q\in J^+(S)\setminus I^+(S) = E^+(S)$.
The other possibility is that the $t_k$ are unbounded, i.e.\ w.l.o.g.\ $t_k \to
\infty$. Since $(M,g)$ is strongly causal, this means that the $\gamma_k$ leave every
compact set, i.e., $\gamma_k(t_k) \to \infty$ (after possibly relabelling $t_k$).
However, as $E^+(S)$ is contained in some compact set, this is not possible.
\end{proof}

\begin{Corollary}
Let the assumptions be as in Proposition \ref{Proposition: submfdnullcompletenessortrappedset}. Then $E^+(S) \cap S$ is achronal, and $E^+(E^+(S) \cap S)$ is compact or $M$ is null geodesically incomplete.
\end{Corollary}
\begin{proof}
    This is shown as in the smooth case, see \cite[Prop.\ 4]{GS} or
    \cite[Prop.\ 4.3]{Seno1} and using that for any $p\in S$ there is some neighbourhood $U$ such that $S\cap U$ is achronal in $U$.
\end{proof}

We now start with our discussion of ``trapped points''. In \cite[Sec.\ 6.3]{GGKS}, a faithful way to generalise the classical condition, which uses Jacobi tensor classes---a tool no longer available already for $C^{1,1}$-metrics---is presented: It uses the mean curvature of spacelike $2$-surfaces given as the level sets of the exponential map that generate the light cone. While the latter tool presently is not at our disposal, the very formulation of the corresponding condition again uses support manifolds and can be adopted verbatim. However, it is now more delicate to derive that the horismos of a trapped point is a trapped set.

\begin{definition}(Trapped points)
\label{definition: trappedset}

\noindent A point $p \in M$ is \textit{future trapped} if for any future pointing null vector $v \in T_pM$ and for any null geodesic $\gamma$ with $\gamma(0) = p$, $\dot{\gamma}(0)=v$, there exists a parameter $t$ and a spacelike $C^2$-submanifold of codimension $m=2$ with $\Tilde{S}\subseteq J^+(p)$, $\gamma(t) \in \Tilde{S}$ and $\mathbf{k}_{\Tilde{S}}(\dot{\gamma}(t)) > 0$.
\end{definition}

\begin{Proposition}
\label{Proposition: trappedpoint}
Let $(M,g)$ be a strongly causal, MNNB $C^1$-spacetime
satisfying the distributional null energy condition. If $p \in M$ is a future
trapped point and $(M,g)$ is null geodesically complete then $E^+(p)$ is compact.
\end{Proposition}
\begin{proof}
Since trapped points are defined by means of codimension-$2$ submanifolds $\Tilde{S}$, we will use  Proposition \ref{Proposition: caseoftrappedsurfaces} for which (only) the distributional null energy condition is needed. The proof is analogous to that of Proposition \ref{Proposition: submfdnullcompletenessortrappedset}.

\textit{$E^+(p)$ is relatively compact:} Suppose $q_j \in E^+(p)$, $q_j \to
\infty$. Let $\gamma_j:[0,t_j] \to M$ be maximising null geodesics connecting $p$
to $q_j$, cf.\ Lemma \ref{Lemma: nullgeodinlightcone}. Then, up to a subsequence, the $\gamma_j$ converge to a $g$-null ray
$\gamma:[0,\infty) \to M$ from $p$ (the domain is $[0,\infty)$ even in $g$-affine
parametrisation since $(M,g)$ is null geodesically complete). By assumption,
there is a $C^2$-spacelike submanifold $\Tilde{S}$ of codimension $2$ and a
parameter $t$ such that $\mathbf{k}_{\Tilde{S}}(\dot{\gamma}(t)) > c > 0$ and we let $b >
1/c$. $\gamma$ is indeed an $\tilde{S}$-ray: If not, it would enter $I^+(\tilde{S}) \subseteq I^+(p)$, 
a contradiction. But by the previous results in this section, $\gamma$ cannot maximise from $\tilde{S}$
past $b$, a contradiction since it is a ray.

\textit{$E^+(p)$ is closed:}
This is shown in exactly the same way as closedness was shown in the proof of Proposition \ref{Proposition: submfdnullcompletenessortrappedset}.
\end{proof}

\section{The main result}\label{sec:mainresult}

 Given the results established so far, the general mechanics of the proof of the Hawking--Penrose theorem remains the same as in the smooth case. There do, however, remain notable differences in the details in this lower regularity
(e.g.\ the use of \cite[Prop.\ 2.13]{G20} in the proof of Theorem \ref{thm: HPcausalityversion}), so we include the full argument.

\begin{Lemma}
\label{Lemma: EminusF}
Let $(M,g)$ be a strongly causal $C^1$-spacetime such that no inextendible null geodesic in $M$ is 
maximising. Let $A \subseteq M$ be achronal such that $E^+(A)$ is compact, and let $\gamma$ be a future
inextendible timelike curve contained in $D^+(E^+(A))^{\circ}$. Then $F:=E^+(A) \cap 
\overline{J^-(\gamma)}$ is achronal and $E^-(F)$ is compact.
\end{Lemma}
\begin{proof}
Due to Lemma \ref{Lemma:achronal_closed} we may assume that $A=\overline{A}$. Since $E^+(A)$ is always achronal and $F\subseteq E^+(A)$, it is also
achronal. By compactness of $E^+(A)$, $F$ is compact. Note that $E^-(F) \subseteq F \cup \partial
J^-(\gamma)$ by similar arguments as in \cite[Lem.\ 9.3.4]{Krie}, so to show compactness of $E^-(F)$,
it suffices to show compactness of $E^-(F) \cap \partial J^-(\gamma)$. 

Suppose now that $v \in TM|_F$ is past pointing causal and let $c_v:[0,a) \to M$ be any past
inextendible, past directed geodesic with $\dot c_v(0)=v$ (recall that in $C^1$-spacetimes we no longer
have unique solvability of the geodesic initial value problem). Then $c_v \subseteq \overline{J^-(\gamma)}$, and $\overline{J^-(\gamma)}=I^-(\gamma) \cup \partial J^-(\gamma)$. We show that $c_v$ meets
$I^-(\gamma)$: Suppose, to the contrary, that $c_v$ is a null geodesic entirely contained in $\partial J^-(\gamma)$. Since
$\gamma$ is a future inextendible timelike curve, we have $J^-(\gamma) = I^-(\gamma)$, hence $c_v$ lies
entirely in $\partial J^-(\gamma)\setminus J^-(\gamma)$.
By Corollary \ref{Corollary: curvesinlightcone}, there is a future directed, future inextendible null
geodesic $\lambda$ starting at $c_v(0)$ entirely in $\partial J^-(\gamma)$ ($\gamma$ is closed in $M$
by Corollary \ref{Corollary: imagesofcausalcurvesclosed} and Lemma \ref{Lemma: nontotalnonpartialimprisonment}). This means that either $c_v \lambda$ is 
an inextendible broken null geodesic, hence not maximising by Lemma \ref{Lemma: maxunbrokengeodesic}, or an
inextendible unbroken maximising null geodesic. By assumption such a geodesic cannot exist, hence
in either case $c_v \lambda$ cannot lie entirely in $\partial J^-(\gamma)$ and so $c_v$ must enter $I^-(\gamma)$.

Thus we have shown that for any $v$ and any geodesic $c_v$ as above there is $t_v = t_v(v,c_v)$ 
(for which $c_v$ is still defined) such that $c_v(t_v) \in I^-(\gamma)$. It remains to show that $E^-(F)\cap \pt J^-(\gamma)$ is compact.

\textit{$E^-(F) \cap \partial J^-(\gamma)$ is relatively compact:} Suppose there are
$q_k \in E^-(F) \cap \partial J^-(\gamma)$ with $q_k \to \infty$. Then there are past-directed
maximising null pre-geodesics $c_k:[0,t_k] \to M$ in $h$-unit speed parametrisation
with $c_k(0) = p_k \in F$ and $c_k(t_k) = q_k$. Since $q_k \to \infty$, also $t_k \to \infty$ 
By compactness, we may assume $p_k \to p \in F$. Since $q_k\to\infty$,
there is a neighbourhood $U$ which almost all $c_k$ leave and we are in a
position to apply the Theorem \ref{Theorem: limitcurvetheorem}. We get
a past-directed inextendible null limit curve $c:[0,\infty) \to M$ from $p$ and a subsequence of
the $c_k$ (w.l.o.g.\ $c_k$ itself) which converges to $c$ locally uniformly. It is
easy to see that all $c_k$ have to lie entirely in $\partial J^-(\gamma)$, hence so does the limit $c$. In particular, $c$ is everywhere
maximising and may be reparametrised as an inextendible geodesic $c:[0,a) \to
M$ entirely contained in $\partial J^-(\gamma)$. This is a contradiction since we
showed before that each past directed null geodesic with initial velocity in $TM|_F$
has to enter $I^-(\gamma)$.

\textit{$E^-(F) \cap \partial J^-(\gamma)$ is closed}: Let $q_j \in E^-(F) \cap \partial J^-(\gamma)$, $q_j
\to q$, where we may assume $q\not\in F$, otherwise there is nothing to prove. Connect $F \ni p_j \to q_j$ via maximising null geodesics. Up to a subsequence, their limit has to be a maximising null geodesic from $p$ to $q$, hence $q \in E^-(F) \cap \partial J^-(\gamma)$, because
otherwise one would again obtain an inextendible maximising null geodesic entirely in $\partial J^-(\gamma)$, which would lead to the
same contradiction as before.
\end{proof}

Now we are ready to give the $C^1$-version of the causal Hawking-Penrose theorem. It is the direct analog of \cite[p.\ 538, Thm.]{HP} and \cite[Thm.\ 7.4]{GGKS}, respectively.

\begin{Theorem}\label{thm: HPcausalityversion}
Let $(M,g)$ be a $C^1$-spacetime. Then the following conditions cannot all hold:
\begin{enumerate}
    \item $(M,g)$ is causal.
    \item No inextendible timelike geodesic in an open globally hyperbolic subset is everywhere maximising.
    \item No inextendible null geodesic is everywhere maximising.
    \item There is an achronal set $A \subseteq M$ such that $E^+(A)$ or $E^-(A)$ is compact.
\end{enumerate}
\end{Theorem}
\begin{proof}
Suppose all of the above conditions hold. Then by Lemma \ref{Lemma: sufficientconditionstrongcausality}, (i) and (iii) imply that $(M,g)$ is strongly causal. W.l.o.g.\ we assume that $E^+(A)$ is compact.
By Lemma \ref{lemma:deplus} there exists a future inextendible (in M) timelike curve $\gamma \subseteq D^+(E^+(A))^\circ$. Set $F:= E^+(A)\cap \overline{J^-(\gamma)}$, which by Lemma \ref{Lemma: EminusF} is achronal and $E^-(F)$ is compact. Once more by Lemma \ref{lemma:deplus} there is some past-inextendible timelike curve $\lambda \subseteq D^-(E^-(F))^\circ$. 
Exactly as in the smooth case, see e.g.\ the second paragraph of\ \cite[Proof of Thm.\ 7.4]{GGKS}, one can show $\gamma \subseteq D^+(E^-(F))^\circ$, by invoking Proposition \ref{Proposition: closurecauchydevelopment} and Lemma \ref{Lemma: futureofhorizon}. So $\gamma, \lambda \subseteq D(E^-(F))^\circ$ which by Proposition \ref{Proposition: globhypcauchydevelopment} is globally hyperbolic. 
Now we pick sequences $p_k=\lambda(s_k)$ and $q_k=\gamma(t_k)$ with $t_k, s_k \nearrow \infty$, leaving every compact set, which is possible because $\gamma$ and $\lambda$ are inextendible curves in $M$ contained in a strongly causal set (Lemma \ref{Lemma: nontotalnonpartialimprisonment}). Also because $\lambda \subseteq J^-(E^-(F)) \subseteq J^-(F)\subseteq J^-(\overline{J^-(\gamma}))$ and because $\gamma$ is timelike, we can choose $p_k \in I^-(q_k)$.

By \cite[Prop.\ 2.13]{G20} there exist maximising timelike geodesics $\tilde\gamma_k$ in $D(E^-(F))^\circ$ from $p_k$ to $q_k$, which must intersect $E^-(F)$, say at $r_k$. By compactness of $E^-(F)$ we can assume $r_k \to r$ and by Corollary \ref{Corollary: Hartmangeodesicversion} there exists an inextendible, causal limit geodesic $\tilde\gamma$, which is also maximising (as a limit of maximisers).

By (iii), $\tilde\gamma$ cannot be null, so it must be timelike. As a limit of curves in $D(E^-(F))^\circ$ it is contained in $\overline{ D(E^-(F))}$. Again using Proposition \ref{Proposition: closurecauchydevelopment} and Lemma \ref{Lemma: futureofhorizon} one can now proceed as in the last step in \cite[Proof of Thm. 7.4]{GGKS}
to arrive at the desired contradiction (to (ii)).
\end{proof}

Finally, we formulate and prove the main result of this work, the analytical Hawking-Penrose theorem for maximally causally non-branching $C^1$-spacetimes.

\begin{Theorem}(The Hawking-Penrose singularity theorem for $C^1$-metrics)
\label{Theorem: HPC1}

\noindent Let $(M,g)$ be a $C^1$-spacetime of dimension $n$ with the following properties.
\begin{enumerate}
    \item $(M,g)$ is causal.
    \item $(M,g)$ satisfies the distributional timelike and null energy conditions.
    \item $(M,g)$ satisfies the distributional genericity condition along any inextendible causal geodesic.
    \item $(M,g)$ is maximally causally non-branching.
    \item $(M,g)$ contains one of the following.
    \begin{enumerate}
        \item a compact, achronal, edgeless set.
        \item a future trapped point.
        \item a future trapped $C^0$-submanifold of codimension $2$.
        \item a future trapped $C^0$-submanifold of codimension $2<m<n$ such that its support submanifolds satisfy the condition in 
        Proposition \ref{Proposition: lightraysfromsubmfd}.
    \end{enumerate}
\end{enumerate}
Then $(M,g)$ is causally geodesically incomplete.
\end{Theorem}

\begin{proof}
Assuming causal geodesic completeness, we would like to deduce a contradiction by using Theorem \ref{thm: HPcausalityversion}.

Note that due to Theorems \ref{Theorem: notimelikelines} and
\ref{Theorem: nonulllines} the assumptions (ii) and (iii) of Theorem \ref{thm: HPcausalityversion}
are satisfied and it remains to establish the existence of a future or past trapped set, i.e.\ 
an achronal set with compact future or past horismos. This will be implied by (v), but we have
to distinguish the different cases:

By Proposition \ref{Proposition: hypersufacecase}, (a) implies the existence of a trapped set.
Note that by Lemma \ref{Lemma: sufficientconditionstrongcausality} (which can be applied by Theorem
\ref{Theorem: nonulllines} ) $(M,g)$ is strongly causal and hence Proposition 
\ref{Proposition: trappedpoint} covers the case (b), whereas Proposition
\ref{Proposition: submfdnullcompletenessortrappedset} covers cases (c) and (d).

\end{proof}

\appendix

\section{Appendix: Causality results in $C^1$-spacetimes}
\label{app:C11causality}

In this section, we collect various results on the causality of $C^1$-Lorentzian metrics that will be needed in the proof of the main result.
Unless stated otherwise, throughout this section $(M,g)$ consists of a connected smooth (Hausdorff and second countable) manifold $M$, and a $C^1$-Lorentzian metric $g$ that is time-orientable (i.e.\ there exists a smooth timelike vector field on $M$). Many of the results in this section hold in even greater generality, namely either for Lorentzian metrics that are merely Lipschitz or even in the general
setting of closed cone structures \cite{Min_closed_cone}. Although we will occasionally note this, in proving the results in this section we will stay close to smooth causality theory (mainly based on the comprehensive reference work \cite{MinguzziLivingReview}) and only make those changes from classical proofs that are 
required due to the absence of certain tools (e.g., convex normal neighbourhoods) in the $C^1$-setting.
\medskip\\
We begin with two elementary results, namely the openness of timelike futures and pasts and the push-up property. Both of these hold for causally plain spacetimes, which include spacetimes with Lipschitz continuous metrics \cite[Cor.\ 1.17]{CG}.
\begin{Lemma}
\label{Lemma: futureopen}
For any $A \subseteq M$, the sets $I^{\pm}(A)$ are open in $M$.
\end{Lemma}
\begin{proof}
See \cite[Prop.\ 1.21]{CG}.
\end{proof}

\begin{Lemma}
\label{Lemma: pushup}
Let $p,q,r \in M$ such that $p \leq q \ll r$ or $p\ll q \leq r$. Then $p \ll r$.
\end{Lemma}
\begin{proof}
See \cite[Lem.\ 1.22]{CG}.
\end{proof}

\begin{remark}
Note that from this one can show as in the smooth case, e.g.\ \cite[Lem. 14.03]{ON83}\footnote{the use of convex sets is not necessary once openness of $I^+$ in 
is established},
that the timelike relation is open, i.e., if $p\ll q$ then there are neighbourhoods 
$U_p, U_q$ of $p$ and $q$ respectively, such that for all $x \in U_p$ and 
all $y\in U_q$, we have $x \ll y$.
\end{remark}

\begin{Lemma}
\label{Lemma: maxunbrokengeodesic}
Let $(M,g)$ be a $C^1$-spacetime and let $\gamma:[a,b] \to M$ be a maximising causal curve. Then $\gamma$ is a (reparametrisation of a) causal geodesic. In particular, it is a $C^2$-curve and has a causal character.
\end{Lemma}
\begin{proof}
See \cite[Thm.\ 1.1]{LLS20} or \cite[Thm.\ 3.3]{SS21}.
\end{proof}

For any point $q \in M$, we denote by $E^+(q):=J^+(q) \setminus I^+(q)$ its \textit{future horismos}. Its \textit{past horismos} $E^-(q)$ is defined analogously.

\begin{Lemma}
\label{Lemma: nullgeodinlightcone}
If $p \in E^+(q)$ then there exists a maximising null geodesic segment from $q$ to $p$.
\end{Lemma}
\begin{proof}
Let $c$ be a future causal curve from $q$ to $p$. Since $p \notin I^+(q)$, $c$ is null and maximising from $q$ to $p$ (because of push-up), hence it is (a reparametrisation of) 
a maximising null geodesic by Lemma \ref{Lemma: maxunbrokengeodesic}.
\end{proof}

\begin{Lemma}
\label{Lemma: basicpropertiesoffutures}
For any subset $A \subseteq M$, we have $I^{\pm}(A) = J^{\pm}(A)^{\circ}$ and $\partial I^{\pm}(A) = \partial J^{\pm}(A)$. 
Moreover, the sets $E^{\pm}(A)$ and $\partial J^{\pm}(A)$ are achronal.
\end{Lemma}
\begin{proof}
This follows by the same proofs as in the smooth case (cf.\ \cite[Thm.\ 2.27]{MinguzziLivingReview}, \cite[Cor.\ 14.27]{ON83}).
\end{proof}

\begin{definition}(Edge)
\label{definition: edge}

\noindent Let $A \subseteq M$ be achronal. The \textit{edge} of $A$, denoted by $\edge(A)$, is the set of all $x \in \overline{A}$ such that for any neighbourhood $U$ of $x$, there is a timelike curve from $I_U^-(x)$ to $I_U^+(x)$ that does not meet $A$.
\end{definition}

\begin{Lemma}
\label{Lemma: achronaledgehypersurface}
Let $(M,g)$ be a $C^1$-spacetime. An achronal set $A \subseteq M$ is a topological hypersurface if and only if $A \cap \edge(A) = \emptyset$. Moreover, $A$ is a closed topological hypersurface if and only if $\edge(A) = \emptyset$.
\end{Lemma}
\begin{proof}
The proofs for smooth metrics (cf.\ \cite[Prop.\ 14.25, Cor.\ 14.26]{ON83}) still hold in $C^1$.
\end{proof}

\begin{Lemma}
\label{Lemma: lightconetophsf}
For any $A \subseteq M$, $\partial J^+(A)$ is a (topologically) closed, achronal topological hypersurface.
\end{Lemma}
\begin{proof}
This follows from the fact that $\partial J^+(A)$ is achronal and $\edge(\partial J^+(A)) = \emptyset$, which is proven in precisely the same way as in the smooth case.
\end{proof}

Recall the notion of a limit maximising sequence of causal curves: A sequence $\gamma_k:[a_k,b_k] \to M$ of future directed causal curves is called \textit{limit maximising} if there is $\varepsilon_k \to 0$ such that
\begin{align*}
    L(\gamma_k) \geq d_g(\gamma(a_k),\gamma(b_k)) - \varepsilon_k,
\end{align*}
where $d_g$ is the time separation function induced by $g$. In particular a sequence of maximising curves is limit maximising. We shall employ the following version of the limit curve theorem:

\begin{Theorem}(Limit Curve Theorem)
\label{Theorem: limitcurvetheorem}

\noindent Let $h$ be a complete Riemannian metric on $M$ and let $\gamma_k:[a_k,b_k] \to M$ be future
directed, $h$-arc length parametrised causal curves with $0 \in [a_k,b_k]$ such that the sequence $\{\gamma_k(0)\}$ has
an accumulation point $y$. Suppose there are $a \leq 0$ and $b \geq 0$ such that $a_k \to a$,
$b_k \to b$. If there is a neighbourhood $U$ of $y$ such that almost all $\gamma_k$ leave $U$,
then a subsequence of $\gamma_k$ converges $h$-uniformly on compact sets to a future
directed causal Lipschitz curve $\gamma:[a,b] \to M$. $\gamma$ is future resp.\ past
inextendible iff $b = \infty$ resp.\ $a=-\infty$. If $\{\gamma_k\}$ is limit maximising, then
$\gamma$ is maximising.
\end{Theorem}
\begin{proof}
This follows from \cite[Thm.\ 3.1]{Minguzzi_LimitCurveThms} and \cite[Thm.\ 1.6]{CG} by the same proof as in \cite[Thm.\ A.6]{GGKS}.
\end{proof}
For a much more general version of the limit curve theorem, valid in closed cone structures, we refer to \cite[Thm.\ 2.14]{Min_closed_cone}. We shall also require the following variant of the limit curve theorem for a sequence of curves with converging past and future endpoints, cf.\ \cite[Thm.\ 3.1(2)]{Minguzzi_LimitCurveThms}.

\begin{Theorem}(Limit Curve Theorem: two converging endpoints)
\label{Theorem: limitcurvetheoremtwoconvergingendpoints}

\noindent Let $\gamma_k:[0,b_k] \to M$ be a sequence of future-directed,
causal curves in $h$-arc length parametrisation connecting $x_k$ to
$z_k$, with $x_k \to x$ and $z_k \to z$. Suppose that $b_k \to b
\in (0,\infty)$. If there is a neighbourhood $U$ of $x$ such that
only a finite number of $\gamma_k$ are entirely contained in $U$,
then there is a Lipschitz future-directed, causal curve $\gamma:[0,b] \to M$
connecting $x$ to $z$ such that a subsequence of the $\gamma_k$
converges to $\gamma$ $h$-uniformly on compact sets. Moreover, if
$\{\gamma_k\}$ is limit maximising, then $\gamma$ is maximising.
\end{Theorem}
\begin{proof}
    This is contained in the statement of \cite[Thm.\ 3.1(2)]{Minguzzi_LimitCurveThms}, whose proof is easily seen to hold for $C^1$-metrics once Theorem \ref{Theorem: limitcurvetheorem} is established.
\end{proof}

\begin{Corollary}
\label{Corollary: curvesinlightcone}
Let $A \subseteq M$. For any point $x \in \partial J^+(A)\setminus \overline{A}$, there is a causal curve $\gamma \subseteq \partial J^+(A)$ with future endpoint $x$ that is either past inextendible and does not meet $\overline{A}$ or has a past endpoint in $\overline{A}$. It is a (reparametrisation of a) maximising null geodesic. If $A$ is closed and $x \notin J^+(A)$, then $\gamma \subseteq \partial J^+(A) \setminus J^+(A)$ and it is past inextendible.
\end{Corollary}
\begin{proof} Using Theorem \ref{Theorem: limitcurvetheorem} and Lemma \ref{Lemma: nullgeodinlightcone}, this follows as in 
\cite[Prop.\ A.7]{GGKS}.
\end{proof}

\begin{Lemma}
\label{Lemma: limitsofmaximizers} 
Let $(M,g)$ be a $C^1$-spacetime and let $\{g_k\}$ be either the sequence $\{\hat{g}_{\varepsilon_k}\}$ or $\{\check{g}_{\varepsilon_k}\}$ with $\varepsilon_k \downarrow 0$. Let $p,q \in M$ and let $\gamma_k$ be a $g_k$-maximising curve from $p$ to $q$, so $\gamma_k$ is in particular a $g_k$-geodesic. Suppose that $\gamma_k$ converge to a curve $\gamma$ in $C^1_{\mathrm{loc}}$. If $\{g_k\} = \{\check{g}_{\varepsilon_k}\}$, suppose in addition that $(M,g)$ is globally hyperbolic and $p \ll_g q$.
Then $\gamma$ is a $g$-maximising geodesic from $p$ to $q$.
\end{Lemma}
\begin{proof}
    The case $\{g_k\} = \{\hat{g}_{\varepsilon_k}\}$ follows
    immediately from \cite[Prop.\ 6.5]{S14}, so we will only
    consider the case $\{g_k\} = \{\check{g}_{\varepsilon_k}\}$
    together with the additional assumptions of $(M,g)$ being
    globally hyperbolic and $p \ll_g q$ in this case. Due to global hyperbolicity, the space $C(p,q)$ of future directed causal curves from $p$ to $q$ (considered up to orientation preserving reparametrisations) is compact with respect to its natural topology (see \cite[Thm.\ 3.2]{S14}). 
    Consequently, the $h$-speeds of curves in $C(p,q)$ are uniformly bounded by
    some constant $C_1 >0$, where $h$ is some complete Riemannian metric on $M$.
    
    Since $(M,g)$ is globally hyperbolic, there is an $L_g$-maximising $g$-geodesic
    $c:[0,1] \to M$ with $c(0)=p$, $c(1)=q$ \cite[Prop.\ 6.4]{S14}. 
    Also, $p \ll_g q$ implies that $c$ is $g$-timelike. By compactness, there is a
    constant $C_2 > 0$ such that $g(\dot{c},\dot{c}) < -C_2$ on $[0,1]$, and hence
    $g_k(\dot{c},\dot{c}) < 0$ for $k$ large, implying that $c$ is $g_k$-timelike
    for such $k$. If we set $\delta_k:=d_h(g,g_k)$ (the $h$-distance of $g,
    g_k$ as in \cite[(1.6)]{CG}), 
    then
    \begin{align*}
        L_g(c) &= \int_0^1 \sqrt{-g(\dot{c},\dot{c})} dt \leq \int_0^1 \sqrt{-g_k(\dot{c},\dot{c}) + C_1^2 \delta_k} dt \leq L_{g_k}(c) + C_1 \sqrt{\delta_k}\\
        &\leq d_{g_k}(p,q) + C_1 \sqrt{\delta_k} = L_{g_k}(\gamma_k) + C_1 \sqrt{\delta_k} \to L_g(\gamma) \qquad (k\to \infty),
    \end{align*}
    so $L_g(\gamma) \geq L_g(c)$, which proves that also $\gamma$ is maximising from $p$ to $q$.
\end{proof}

\begin{Lemma}\label{Lemma:nulllimit}
    Let $\gamma_k \subseteq \overline{J^+(p)}$ be causal curves converging locally
    uniformly to a causal curve $\gamma$ with future endpoint $q$. If $d_g(p,q)=0$, then
    $\gamma$ is a maximising null geodesic.
\end{Lemma}
\begin{proof}
    Clearly $\gamma \subseteq \overline{J^+(p)}$. If $\gamma$ would meet $I^+(p)$, then 
    there would exist some $r\in I^+(p)\cap J^-(q)$ and using push-up
    we could conclude that $q \in I^+(p)$, contradicting $d_g(p,q)=0$. Hence 
    $\gamma \subseteq \partial J^+(p)$, which is achronal and $\gamma$ is a 
    maximising null geodesic by Lemma \ref{Lemma: maxunbrokengeodesic}.
\end{proof}

\begin{Lemma}
\label{Lemma: nulllines}
Let $g_k=\check g_{\eps_k}$ and let $\gamma_k \subseteq \overline{J^+(p)}$  be $g_k$-null geodesics converging in $C^1_{\mathrm{loc}}$ to a $g$-geodesic $\gamma$ that ends in $q$. Assume that $d_g(p,q)=0$.
Then $\gamma$ is a maximising $g$-null geodesic.
\end{Lemma}
\begin{proof}
This follows immediately from Lemma \ref{Lemma:nulllimit} upon noting that any $g_k$-null geodesic is $g$ causal.
\end{proof}

\begin{Lemma}\label{Lemma:achronal_closed}
Let $A$ be achronal. Then so is $\overline{A}$. Moreover, if $E^+(A)$ is closed, then $E^{\pm}(A) = E^{\pm}(\overline{A})$.
\end{Lemma}
\begin{proof}
This follows by the same proof as in \cite[Lem.\ A.8]{GGKS}.
\end{proof}

\begin{Lemma}
\label{Lemma: baseofglobhypnbhds}
Let $(M,g)$ be a $C^1$-spacetime. Then every point in $M$ has a neighbourhood base of globally hyperbolic neighbourhoods.
\end{Lemma}
\begin{proof}
    Fix any smooth, time-orientable Lorentzian metric $\hat{g}$ with $g \prec \hat{g}$. For $x \in M$, \cite[Thm.\ 2.7]{MinguzziLivingReview} gives a neighbourhood base $V_m$ of neighbourhoods of $x$ that are globally hyperbolic with respect to $\hat{g}|_{V_m}$. Any Cauchy hypersurface for $\hat g$ in $V_m$ is also a Cauchy
    hypersurface for $g$, so the claim follows from \cite[Thm.\ 5.7]{S14}.
\end{proof}

\begin{Lemma}\label{Lemma:local_Cauchy_hypersurface} Let $S$ be a spacelike ($C^2$-) hypersurface and let $p\in S$. Then there exists a neighbourhood
$V$ of $p$ in $M$ such that $V\cap S$ is a Cauchy hypersurface in $V$.
\end{Lemma}
\begin{proof}
The proof of the analogous result \cite[Lemma A.25]{hawkingc11} also works in the $C^1$-setting.  
\end{proof}

\begin{definition}(Cauchy development and Cauchy horizon)

\noindent Let $A$ be an achronal set. The \textit{future Cauchy development} of $A$ is the set
\begin{align*}
    D^+(A):= \{x \in M: \text{ every past inextendible causal curve through }x\text{ meets }A\}.
\end{align*}
The \textit{future Cauchy horizon} of $A$ is defined as
\begin{align*}
    H^+(A):=\overline{D^+(A)} \setminus I^-(D^+(A)).
\end{align*}
The \textit{past Cauchy development and past Cauchy horizon} are defined analogously.
\end{definition}
The first part of the following result is partly a $C^1$-analogue of \cite[Lem.\ 14.51]{ON83}, where
it is proved using convex normal neighbourhoods, a tool that we do not have at our disposal in
the $C^1$-setting. The use of such neighbourhoods can, however, be avoided, as our arguments below illustrate.

\begin{Proposition}
\label{Proposition: closurecauchydevelopment}
Let $A$ be a closed, achronal set. Then
\begin{align*}
    \overline{D^+(A)} = \{x \in M: \text{every p.i. t.l. curve through }x\text{ meets }A\}
\end{align*}
and
\begin{align*}
    \partial D^+(A) = A \cup H^+(A).
\end{align*}
\end{Proposition}
\begin{proof}
Let $T$ be the set on the right hand side in the first equation.

\noindent $\overline{D^+(A)}\subseteq T$: Suppose there is a point $p \in \overline{D^+(A)}\setminus T$. Then there is some past inextendible timelike (future directed) curve $\alpha:(a,0] \to M$ with $\alpha(0) = p$ not meeting $A$. Let $p_k \in D^+(A)$, $p_k \to p$ and let $U_k$ be open neighbourhoods of $p$ with $U_k \to \{p\}$. By choosing subsequences, we may assume that $p_k \in I^+_{U_k}(\alpha(-1/k))$ and $\alpha|_{[-1/k,0]} \subseteq U_k$ for all $k$.\\
Construct now past inextendible causal curves $\alpha_k$ by connecting $\alpha(-1/k)$ to $p_k$ in a timelike way such that these curves stay in $U_k$, and let the rest of $\alpha_k$ (i.e.\ the part to the past of $\alpha(-1/k)$) be the original curve $\alpha$. Then the $\alpha_k$ are all past inextendible timelike, and since $p_k \in D^+(A)$ they must meet $A$ at some $a_k$. Since $\alpha$ does not meet $A$, the $a_k$ have to lie on the replaced piece. But these pieces are entirely in $U_k$, which go to $\{p\}$, so we see that $a_k \to p$. Since $A$ is closed, we have $p \in A$, which is absurd since $p \notin T \supseteq A$ by assumption.

\noindent $\overline{D^+(A)}\supseteq T:$ Suppose $q \notin \overline{D^+(A)}$, and let $r \in I^-_{M\setminus \overline{D^+(A)}}(q)$. Since in particular $r \notin D^+(A)$, there is a past inextendible causal curve $\alpha$ from $r$ missing $A$. Precisely as in the smooth case \cite[Lem.\ 14.30(1)]{ON83}, noting that pushup still holds in $C^1$), there is a past inextendible timelike curve from $q$ not meeting $A$. Thus $q \notin T$.

Finally, $\partial D^+(A) = A \cup H^+(A)$ follows exactly as in the smooth case (cf.\ \cite[Lem.\ 14.52]{ON83}).
\end{proof}
We note that the previous result in fact remains true even for locally Lipschitz proper cone structures \cite[Thm.\ 2.35, 2.36]{Min_closed_cone}.

\begin{Lemma}
\label{Lemma: horizonclosedachronal}
Let $A$ be achronal. Then $H^+(A)$ is a closed, achronal set. 
\end{Lemma}
\begin{proof} The proof of \cite[Prop.\ 3.15]{MinguzziLivingReview} carries over unchanged to $C^1$ metrics.
\end{proof}

\begin{Lemma}
\label{Lemma: meetingtimelikepast}
Let $A$ be closed and achronal. If $x \in D^+(A) \setminus H^+(A)$, then every past inextendible causal curve through $x$ meets $I^-(A)$.
Moreover, if $x\in D(A)^\circ$, then every future (resp.\ past) inextendible  future (resp.\ past) directed causal curve emanating from $x$
intersects $I^+(A)$ (resp.\ $I^-(A)$).
\end{Lemma}
\begin{proof}
This follows from \cite[Prop.\ 3.27]{MinguzziLivingReview} and \cite[Prop.\ 3.42]{MinguzziLivingReview}, which still hold for $C^1$-metrics. 
\end{proof}

\begin{Lemma}
\label{Lemma: meetinghorizon}
Let $A$ be closed, achronal. Let $x \in J^+(A) \setminus D^+(A)$ or $x \in I^+(A) \setminus D^+(A)^{\circ}$. Then every causal curve from $x$ to $A$ meets $H^+(A)$.
\end{Lemma}
\begin{proof}
    The proof for $C^{1,1}$-metrics \cite[Lem.\ A.12]{GGKS} still holds in $C^1$ because of Lemma \ref{Lemma: futureopen}, Lemma \ref{Lemma: pushup}, 
    and Proposition \ref{Proposition: closurecauchydevelopment}.
\end{proof}

Using these results, one establishes the following formula for the timelike future of the Cauchy horizon of a closed, achronal set precisely as in the $C^{1,1}$-case, cf.\ \cite[Lem.\ A.13]{GGKS}.

\begin{Lemma}
\label{Lemma: futureofhorizon}
Let $A$ be closed and achronal. Then $I^+(H^+(A)) = I^+(A) \setminus \overline{D^+(A)}$.
\end{Lemma}

\begin{Lemma}
\label{Lemma: edgeofcauchyhorizon}
Let $A$ be closed and achronal. Then $\edge(H^+(A)) \subseteq \edge(A)$.
\end{Lemma}
\begin{proof} Based on Lemmas \ref{Lemma: meetinghorizon} and \ref{Lemma: futureofhorizon}, the proof of \cite[Lem.\ A.14]{GGKS} carries over to the $C^1$-setting.
\end{proof}

\begin{Lemma}
\label{Lemma: horizonoflightcone}
Let $A$ be an achronal set, then $H^+(\partial J^+(A))$ is a closed, achronal, topological hypersurface.
\end{Lemma}
\begin{proof}
$H^+(\partial J^+(A))$ is closed and achronal by Lemma \ref{Lemma: horizonclosedachronal}, and by the previous result, combined with
Lemmas \ref{Lemma: achronaledgehypersurface} and \ref{Lemma: lightconetophsf},
$\edge(H^+(\partial J^+(A)))\subseteq \edge(\partial J^+(A)) = \emptyset$, so the claim follows from Lemma \ref{Lemma: achronaledgehypersurface}.
\end{proof}

\begin{Lemma}
\label{Lemma: horizonofclosureofhorismos}
Let $A$ be closed, achronal. Then $H^+(\overline{E^+(A)})\subseteq H^+(\partial J^+(A))$.
\end{Lemma}
\begin{proof}
The proof of \cite[Lem.\ A.16]{GGKS} carries over, using Corollary \ref{Corollary: curvesinlightcone}, Proposition \ref{Proposition: closurecauchydevelopment} 
and Lemmas \ref{Lemma: meetinghorizon}, \ref{Lemma: futureofhorizon}.
\end{proof}

\begin{Proposition}
\label{Proposition: globhypcauchydevelopment} 

Let $(M,g)$ be a $C^1$-spacetime and let $A\subseteq M$ be a closed, achronal set. Then the interior of its Cauchy development, i.e.\ the set $D(A)^{\circ} = (D^+(A) \cup D^-(A))^{\circ}$, if nonempty, is a globally hyperbolic $C^1$-spacetime when considered with the induced metric.
\end{Proposition}
\begin{proof} Based on Theorem \ref{Theorem: limitcurvetheorem} and Lemma \ref{Lemma: meetingtimelikepast}, this follows from the same
proof as in \cite[Thm.\ 3.45]{MinguzziLivingReview}.
\end{proof}

\begin{definition}(Strong causality)
\label{definition: strongcausality}

\noindent A $C^1$-spacetime $(M,g)$ is called \textit{strongly causal at} $p \in M$ if for any neighbourhood $U$ of $p$ there is a neighbourhood $V$ of $p$ with $V \subseteq U$ such that any causal curve starting and ending in $V$ is contained in $U$. $(M,g)$ is called \textit{strongly causal} if it is strongly causal at every point.
\end{definition}

\begin{Lemma}
\label{Lemma: equivalentdefofstrongcausality}
$(M,g)$ is strongly causal at $p$ if and only if for any neighbourhood $U$ of $p$ there is a neighbourhood $V$ of $p$ with $V \subseteq U$ such that no causal curve leaving $V$ ever returns.
\end{Lemma}
\begin{proof}
    This follows from \cite[Lem.\ 3.21]{ladder}, which continues to hold in the $C^1$-case.
\end{proof}

\begin{Lemma}
\label{Lemma: sufficientconditionstrongcausality}
Let $(M,g)$ be chronological and suppose there are no inextendible maximising null geodesics in $M$. Then $(M,g)$ is strongly causal.
\end{Lemma}
\begin{proof} Using Theorem \ref{Theorem: limitcurvetheorem}, this follows by the same argument as \cite[Lem.\ A.19]{GGKS}.
\end{proof}

\begin{Lemma}
\label{Lemma: nontotalnonpartialimprisonment}
A strongly causal $C^1$-spacetime $(M,g)$ is non-totally and non-partially imprisoning: No future or past inextendible causal curve is contained in a compact set and no future or past inextendible causal curve returns to a compact set infinitely often.
\end{Lemma}
\begin{proof}
    The classical proof for the smooth case (see \cite[Lem.\ 14.13]{ON83}) holds even for continuous metrics.
\end{proof}

\begin{Corollary}
\label{Corollary: imagesofcausalcurvesclosed}
Let $(M,g)$ be a non-totally and non-partially imprisoning $C^1$-spacetime. Then the image of every inextendible causal curve is closed in $M$.
\end{Corollary}
\begin{proof}
    Let $\gamma$ be any inextendible causal curve. Since $M$ is a manifold, the topology on $M$ is compactly generated. Let $K \subseteq M$ be any compact subset. By non-total and non-partial imprisonment, $\gamma \cap K$ is a finite union of closed segments of $\gamma$, hence $\gamma \cap K$ is closed in $K$. Since $K$ was arbitrary, $\gamma$ is closed in $M$.
\end{proof}

\begin{Lemma}\label{Lemma:stronglycausalHplus}
Let $(M,g)$ be a strongly causal $C^1$-spacetime and let $A \subseteq M$ be a closed, achronal subset. If $H^+(\overline{E^+(A)})$ is nonempty, then it is noncompact.
\end{Lemma}
\begin{proof}
The proof of the smooth case goes through, because all the necessary preliminary results continue to hold in $C^1$ (see the outline in \cite[Lem.\ 7.1]{GGKS}).
\end{proof}

\begin{Lemma}\label{lemma:deplus}
Let $(M,g)$ be a strongly causal $C^1$-spacetime. Let $A \subseteq M$ be an achronal subset such that $E^+(A)$ is compact. Then there is a future inextendible timelike curve $\gamma$ contained in $D^+(E^+(A))^{\circ}$.
\end{Lemma}
\begin{proof}
The proof of the smooth case goes through in $C^1$, see e.g.\ \cite[Lem.\ 2.12]{HP}, or also \cite[Lem.\ 9.3.3]{Krie}, combined with 
Lemma \ref{Lemma:achronal_closed}.
\end{proof}

\begin{Proposition}
\label{Proposition: hypersufacecase}
Let $(M,g)$ be a $C^1$-spacetime and let $A \subseteq M$ be achronal and edgeless. Then $E^+(A) = A$, in particular if $A$ is compact so is $E^+(A)$.
\end{Proposition}
\begin{proof}
This was shown e.g.\ in \cite[Cor. 2.145]{MinguzziLivingReview} and the proof carries over to $C^1$ spacetimes.
\end{proof}

	\section{Appendix: Extending vector fields in $C^1$-spacetimes uniformly}\label{app:b}
	When applying the genericity condition (Definition \ref{definition: strongdistributionalgenericity})
	we need to extend vectors or vector fields along curves to neighbourhoods, 
	due to the fact that distributional (or also $L^\infty$) curvature quantities
	are not well defined on points or along curves (as these are sets of measure zero).
	It is rather straightforward to extend a vector	field along a curve to a neighbourhood. For $C^1$-metrics one can still use parallel transport along certain curves spanning an open set. Alternatively, one can also use ``cylindrically constant'' (see below) extensions in coordinates around a (part of a) curve. 
	
	In our case, we need to simultaneously extend vector fields along a sequence of converging curves. To be more precise, let $\gamma_k$ be a sequence converging 
	in $C^1$ to some $\gamma$ and let $V_k$ be vector fields along $\gamma_k$ and $V$ along $\gamma$ such that $V_k \to V$ in $C^1$ as $k\to \infty$.
	Then given a point on $\gamma$, does it possess a neighbourhood $U$ such that all $V_k$ can be
	extended to vector fields on $U$ and such that their extensions converge in $C^1$ to the extension of $V$? 
	
	To begin with, we consider the situation of extending vector fields around a point of a curve to a neighbourhood by 
	cylindrically constant extension in coordinates. 

	\begin{Lemma}\label{lem:A1} (Cylindrically constant extension of vector fields on curves)
	
	\noindent Let $\gamma:I\to M$ be a $C^1$-curve which is regular at $0$, i.e.\ $\dot{\gamma}(0)\neq 0$ and set $p:=\gamma(0)$. 
	Then there exists a chart $(U,(x^1,\dots,x^n))$ around $p$ with the 
	following property: For any $C^1$ vector field $V$ along $\gamma$ there exists a  $C^1$ extension 
	$\tilde V$ of $V$ to $U$ such that, in these coordinates, $\tilde V$ is independent of $x^2,\dots,x^n$.
	\end{Lemma}
\begin{proof}
	Since the claim is local we may assume that $M=\R^n$, that $p=0$, and that $\gamma=(\gamma^1,\dots,\gamma^n)$
	with $\gamma_1'(t) > c > 0$ for some $c$ and all $t\in (a,b)$, where $a <0<b$. Thus $\gamma^1$ is a $C^1$-diffeomorphism from $(a,b)$ onto some interval $J$.
	Then we set $U:= J \times B^{n-1}_R(0)$, where $R>0$ is such that $\gamma((a,b))\sse U$. Given a $C^1$-vector field $V$ along $\gamma$, i.e., 
	a $C^1$-map $V:(a,b)\ni t \mapsto (\gamma(t),v(t))$ with $v\in C^1((a,b),\R^n)$, set $\tilde V: U \to U\times \R^{n}$,
	\[
	\tilde V(x) := (x,v((\gamma^1)^{-1}(x_1)))
	\]
	to obtain the desired extension.
\end{proof}

Now that we have set up a way of extending vector fields in a cylindrically constant fashion, we can 
implement an analogous procedure uniformly on a sequence of curves.

\begin{Lemma}\label{lem:A2}
	Let $\gamma_k:[-1,1]\to M$ be a sequence of $C^1$-curves converging in $C^1([-1,1])$ to the regular
	$C^1$-curve $\gamma:[-1,1]\to M$ and let $V_k$ and $V$ be $C^1$ vector fields along
	$\gamma_k$ and 	$\gamma$, respectively. Further, let $V_k\to V$ in $C^1$, i.e., both $V_k\to V$ and $V_k'\to V'$ in $TM$,
	uniformly on $[-1,1]$.
	Then there exists some open neighbourhood $U$ of $\gamma(0)$ and
	$C^1$ extensions $\tilde V_k$ of $V_k$ and $\tilde V$ of $V$ to $U$
	such that $\tilde V_k \to \tilde V$ in $C^1_{\mathrm{loc}}(U)$.
\end{Lemma}
\begin{proof} Again we may assume that $M=\R^n$. As in the previous Lemma we write
$\gamma=(\gamma^1,\dots,\gamma^n)$ and can assume without loss of generality that $\gamma(0)=0$, and that  $(\gamma^1)'(t) > c > 0$
for some $c$ and all $t\in (a,b)$, where $a <0<b$. Due to the assumption on the $C^1$-convergence of the $\gamma_k$,
we may additionally suppose that the same inequality holds on $(a,b)$ for each $(\gamma^1_k)'$. Thus each $\gamma^1_k$
and $\gamma^1$ itself are $C^1$-diffeomorphisms from $(a,b)$ onto their respective image. In addition (restricting to $n$ large if necessary), 
there exists a nontrivial interval $J$ around $0$ that is contained in $\bigcap_k \gamma^1_k((a,b)) \cap \gamma((a,b))$. 
Then we set $U:= J \times B^{n-1}_R(0)$, 
where $R>0$ is chosen such that $\bigcup_k \gamma_k((a,b)) \cup \gamma((a,b))\sse U$. We can write $V$ 
in the form $t \mapsto (\gamma(t),v(t))$ with $v\in C^1((a,b),\R^n)$, and analogously $V_k(t) = (\gamma_k(t),v_k(t))$,
with $v_k \to v$ in $C^1$. For $x\in U$ we set
\[
\tilde V(x) := (x,v((\gamma^1)^{-1}(x^1))) \qquad \tilde V_k(x) := (x,v_k((\gamma_k^1)^{-1}(x^1))),
\]
giving $C^1$-extensions of $V$ resp.\ $V_k$ to $U$. By \cite[Cor.\ 1]{BDF91} we have that $(\gamma_k^1)^{-1}$ converges to 
$(\gamma^1)^{-1}$ locally uniformly, and that the same is true for the first derivatives. Consequently, 
$\tilde V_k \to \tilde V$ in $C^1_{\mathrm{loc}}(U)$.
\end{proof}
As the proof shows, if the $V_k$ converge to $V$ merely in $C^0_{\mathrm{loc}}$ then one can still find $C^1$-extensions to $U$ that
also converge in $C^0_{\mathrm{loc}}$.

\vskip1em

\noindent{\em Acknowledgements.} This work was supported by project P 33594 of the Austrian Science Fund FWF and by the Uni:Docs program of the University of Vienna. We thank Eduard A.\ Nigsch for helpful discussions.

\end{document}